\documentclass{article}

\usepackage{hyperref}       
\usepackage{tikz}
\usepackage{multirow}
\usepackage{url}            
\usepackage{booktabs}       
\usepackage{amsfonts}       
\usepackage{nicefrac}       
\usepackage{microtype}      
\hypersetup{colorlinks,linkcolor={blue},citecolor={blue},urlcolor={blue}}
\usepackage{natbib}
\usepackage{amsfonts}
\usepackage{graphicx,amssymb,mathrsfs,amsmath,color,fancyhdr}
\usepackage{amsthm}
\usepackage{cases}
\usepackage{mathbbol}
\usepackage{geometry}
\usepackage{cite}
\usepackage{booktabs}
\usepackage{array}
\usepackage{amssymb}
\makeatletter

\newcommand{\Rmnum}[1]{\expandafter\@slowromancap\romannumeral #1@}
\makeatother

\newtheorem{theorem}{Theorem}
\newtheorem{lemma}{Lemma}

\newtheorem{proposition}{Proposition}
\newtheorem{assumption}{Assumption}
\newtheorem{corollary}{Corollary}

\usepackage{caption}
\usepackage{amsfonts}

\usepackage{algorithm}
\usepackage{algpseudocode}
\allowdisplaybreaks
\usepackage{rotating}
\usepackage{color}

\usepackage[all]{xy}
\usepackage{bm}
\usepackage{pifont}
\usepackage{stmaryrd}
\usepackage{pgf,tikz}
\usepackage{tikz}
\usetikzlibrary{shapes,arrows,positioning,calc}
\usetikzlibrary{fadings}
\usetikzlibrary{patterns}
\usetikzlibrary{shadows.blur}
\usetikzlibrary{shapes}

\usepackage{latexsym,amsopn,amstext,amsthm,amsxtra,bm,calc,ifpdf}
\usepackage{enumerate}
\usepackage{subfigure}
\usepackage{listings}
\usepackage{enumitem}

\newcommand{\indep}{\perp \!\!\! \perp}
\usepackage[flushleft]{threeparttable}
\usepackage{enumitem}
\usepackage{mathrsfs}

\usepackage{arydshln}

\newcommand{\wyy }{\color{black}} 
\newcommand{\purple}{\color{black}}

\def\influence function{{\rm  influence function}}
\def\T{{\mathrm{\scriptscriptstyle T} }}

\usepackage[toc,page]{appendix}
\usepackage{authblk}
\hypersetup{
	colorlinks=true,
	linkcolor=black
}

\hypersetup{
    colorlinks=true,
    linkcolor=blue,   
    citecolor=blue,   
    filecolor=blue,   
    urlcolor=blue     
}

\usepackage[symbol]{footmisc}
 
\usepackage{xr}

\usepackage{setspace}
\usepackage{caption}
\def\T{{ \mathrm{\scriptscriptstyle T} }}

\title{\bf Selecting among Missingness Models for Sequential Outcomes with Nonignorable Nonresponse}

\usepackage{authblk}

\author[1]{Yingying Wang\thanks{These authors contributed equally.}}
\author[2]{Yuan Liu\thanks{These authors contributed equally.}}
\author[1]{Shanshan Luo}

\affil[1]{School of Mathematics and Statistics, Beijing Technology and Business University}
\affil[2]{Mathematical Institute, Leiden University}

\begin{document}

\date{}
\maketitle

\begin{abstract} 
Sequential outcomes in longitudinal studies and multi-wave surveys may be missing not at random at both earlier and later occasions. We study graphical models in which at least one outcome is self-censoring and the response indicator for a later outcome may depend on either the earlier response indicator or the realized earlier outcome. These restrictions define two candidate families under which the relevant full-data distributions are identifiable and whose observed-data models overlap; graphs containing both dependencies form a broader class outside the prespecified comparison. For each candidate family, we establish identification of the full-data distribution under rank or completeness conditions and develop likelihood-based estimation. We then propose a two-stage Vuong-type procedure. The first stage determines whether the candidate models are observationally distinguishable; only after distinguishability is established does the second stage compare their Kullback--Leibler divergences from the true observed-data distribution. We also show that ordinary Wald inference remains asymptotically valid for the selected model-specific functional when the selected model has a fixed positive expected log-likelihood advantage. Simulations evaluate the two-stage procedure across graph classes. We finally apply the procedure to compare candidate missingness models in the Job Corps data and perform downstream functional estimation under the selected model.
\end{abstract}

\section{Introduction}
\label{sec:intro}

Missing data are a pervasive challenge in longitudinal studies and multi-wave surveys. The problem becomes especially difficult when nonresponse is missing not at random (MNAR), because the probability that a value is observed may depend on the unobserved value itself or on other latent factors \citep{2002Statistical}. Consequently, standard complete-case or missing-at-random analyses can distort the joint distribution of key variables, thereby biasing any downstream scientific conclusions. Existing work has developed important identification and estimation strategies for MNAR covariates or outcomes, including approaches based on instrumental variables, empirical likelihood, conditional independence restrictions, and  semiparametric missingness models \citep{Ding2014,Yang2019,yang2014estimation,d2010new,wang2014instrumental,Miao01102016,sun2018semiparametric}. Much of this literature, however, focuses on settings where only a single substantive variable is subject to missingness, while the remaining key variables are fully observed.

Sequential outcomes create a different and common practical difficulty. In longitudinal studies, an earlier outcome and a later outcome are measured at different time points, and both may be subject to nonignorable nonresponse \citep{glymour2006using,MA200324,tian2015missing,Nabi2026Causal,Zuo2025}. For example, an early educational, employment, or health status may be unavailable because a respondent does not complete an intermediate survey, while a later income or health outcome may be missing because of attrition or follow-up nonresponse. The two nonresponse processes are rarely independent: failure to observe the first outcome may itself signal reduced willingness to participate later, or the realized first outcome may affect whether the later outcome is measured. These mechanisms lead to different observed-data distributions, different identification arguments, and potentially different conclusions when summarizing the recovered full-data joint distribution.

Graphical missing data models provide a natural framework for representing such multivariate missingness structures. \citet{MA200324} used directed acyclic graphs to study nonignorable nonresponse of binary outcomes in longitudinal studies and proposed graphical criteria for model identifiability under a first-order Markov structure. Adopting a counterfactual perspective, \citet{Nabi2026Causal} utilized missing-data DAGs to encode restrictions between potential outcomes and response indicators, thereby clarifying the identifiability of the complete-data distribution. \citet{li2023self} developed chain-graph self-censoring models for multivariate nonignorable missing data and established parameter recoverability through completeness conditions. Expanding these approaches to complex causal mechanisms, \citet{Zuo2025} recently investigated mediation analysis where both the mediator and the outcome are missing not at random. While these contributions demonstrate that graphical representations are invaluable for encoding full-data independence restrictions and deriving identification results, a critical practical challenge remains: when several plausible missingness graphs are available, which one should be used?

This question is typically not addressed by identification theory alone \citep{Ding2014,Yang2019,d2010new,wang2014instrumental,Miao01102016,sun2018semiparametric,Zuo2025}. While an identification result guarantees that the full-data distribution or its relevant components can be recovered assuming a postulated missingness mechanism is correct, it offers no guidance on whether that mechanism is better supported by the observed data than a competing identifiable model. In applied MNAR analyses, uncertainty regarding the missingness mechanism is conventionally handled via sensitivity analysis, pattern-mixture modeling, or tipping-point analysis \citep{rosenbaum1987sensitivity,Scharfstein2021}. Although highly valuable for assessing how scientific conclusions fluctuate across varying assumptions, these approaches do not provide a formal statistical test to distinguish among competing missingness graphs. Furthermore, they cannot rank the candidate mechanisms based on their likelihood-based closeness to the underlying data-generating process.

Model selection offers a complementary perspective for the candidate models. Since \citet{Akaike1998} introduced AIC and \citet{schwarz1978estimating} proposed BIC, Kullback--Leibler-based criteria have become standard tools for comparing fitted models. Such criteria are useful summaries of relative fit, but they do not directly test whether the observed difference between two candidate models is statistically meaningful. \citet{Vuong1989} developed a likelihood ratio framework for non-nested and overlapping models, allowing one to test whether two models have equal Kullback--Leibler divergence from the true observed-data distribution. This feature is particularly relevant for graphical MNAR models because two missingness graphs can be distinct while still sharing a common observed-data submodel. In that case, a model selection procedure must first determine whether the graphs are distinguishable before deciding which graph is favored. The primary target is recovery of the relevant full-data joint distribution, from which substantive quantities, such as contrasts in the later outcome across levels of the earlier outcome, can be obtained as functionals.

In this paper, we focus on two scientifically distinct missingness graphs for later-outcome nonresponse: dependence on whether the earlier outcome was observed and dependence on the realized earlier outcome. In the first, the response indicator for the first outcome directly affects the response indicator for the second outcome. In the second, the first outcome itself affects second outcome nonresponse, while the two response indicators are conditionally independent given the substantive variables. For each graph, we establish nonparametric identification of the relevant components of the full-data joint distribution under rank or completeness conditions. We then develop maximum likelihood estimation using the EM algorithm and propose a two-stage Vuong-type procedure for selecting between the candidate missingness graphs. The first stage tests whether the two candidate graphs are observationally distinguishable. Conditional on distinguishability, the second stage compares their Kullback--Leibler divergences from the observed-data distribution. We derive the asymptotic distributions of the resulting statistics under regularity conditions that allow for MNAR structure and possible model misspecification.

Our main contributions are threefold. First, we organize the candidate graphs into two identifiable families, their overlap, and a broader class outside the main comparison. For each candidate family, we establish identification of the relevant full-data distribution under rank or completeness conditions and develop likelihood-based estimation. Second, we develop a two-stage model selection procedure for comparing the two identifiable families. 
The first stage tests whether their fitted observed-data likelihoods are distinguishable; the second stage compares their Kullback--Leibler divergences from the data-generating distribution only after the first stage rejects indistinguishability. 
Third, we study inference for smooth functionals of the recovered full-data distribution after a candidate model has been selected. Ordinary Wald inference remains asymptotically valid when one model has a fixed likelihood advantage.  

The remainder of this paper is organized as follows. Section~\ref{sec:com_ide} introduces the notation and the sequential missing-data setup. Section~\ref{sec:mis_ide} presents the candidate missingness graphs and states two main identification results. Section~\ref{sec:model_selection} develops likelihood-based estimation and the two-stage model selection procedure. Section~\ref{sec:post_selection} studies inference for a full-data functional after a missingness model has been selected. Section~\ref{sec:sim} reports simulation studies, Section~\ref{sec:app} presents the Job Corps application, and Section~\ref{sec:discussion} concludes with a discussion of limitations and future directions.

\section{Framework and Notation}\label{sec:com_ide}

Throughout this paper, we assume there are $n$ individuals, independently and identically sampled from a superpopulation of interest. For each unit $i$, let $Y_{1i} \in \{0,1\}$ denote a binary outcome measured at an earlier time point, and let $Y_{2i}$ denote a subsequent outcome, which may be continuous or discrete. Let $X_i$ denote a vector of fully observed pre-treatment covariates available prior to the two sequential outcomes. We omit the subscript $i$ throughout the following discussion for notational simplicity.


In longitudinal studies and multi-wave surveys,  $Y_1$ and $Y_2$ are frequently missing at their respective measurement occasions. For instance, an intermediate qualification or other early follow-up outcome may be unreported at an earlier wave, while later outcomes such as income or health status may be missing because of attrition or nonresponse \citep{haneuse2016learning,wells2013strategies,haneuse2021assessing,kolamunnage2016modelling,Zuo2025}. Such sequential missingness patterns substantially complicate distributional recovery, as naive analyses that restrict attention to the fully observed subset may introduce selection bias and yield inconsistent estimates of downstream functionals.

To characterize such missingness patterns, define the response indicators $R_1 \in \{0,1\}$ and $R_2 \in \{0,1\}$, where $R_1=1$ indicates that the first outcome $Y_1$ is observed and $R_1=0$ otherwise, and likewise $R_2=1$ indicates that the second outcome $Y_2$ is observed. {Let $O_i=(X_i,R_{1i},R_{1i}Y_{1i},R_{2i},R_{2i}Y_{2i})$ denote the observed-data vector.}
When both the first and second outcomes are subject to nonignorable
missingness, the central identification task is to recover the relevant
full-data joint distribution from incomplete sequential measurements.
Let $F$ denote this full-data distribution. For a prespecified smooth map
$F\mapsto h_F(x)$, define the population functional
\begin{align}\label{eq:general_functional}
	\tau(F)
	=E_F\{h_F(X)\}
	=\int h_F(x)\,dF_X(x),
\end{align}
where $F_X$ is the marginal distribution of the fully observed covariates.
For example, setting $
h_F(x)
=
E_F(Y_2\mid Y_1=1,x)-E_F(Y_2\mid Y_1=0,x)$ 
yields the population-standardized contrast
\begin{align}
	\tau_{\mathrm{SC}}(F)
	&=
	E_F\!\left\{
	E_F(Y_2\mid Y_1=1,X)
	-
	E_F(Y_2\mid Y_1=0,X)
	\right\}.
	\label{eq:standardized_contrast}
\end{align}
Under suitable causal identification assumptions, including no unmeasured
confounding \citep{Rubin1983}, this contrast can be interpreted as the causal effect of
$Y_1$ on $Y_2$. Other smooth estimands arising in causal mediation
analysis \citep{pearl2014interpretation}, including natural direct and indirect effects, are discussed in
Appendix~\ref{subsec:mediation_functionals}. 


\section{Identification under Competing Missingness Models}
\label{sec:mis_ide}

To establish identification across competing missingness models, we first formalize the candidate graphical structures. Figure~\ref{fig:graph_catalogue} summarizes these graphs, which share a core premise: at least one outcome is self-censoring (i.e., containing the edge $Y_1 \to R_1$ or $Y_2 \to R_2$). Self-censoring is the central nonignorability feature studied here because it directly links an outcome to its own observation status. Requiring at least one such link ensures our comparison focuses on substantively distinct nonignorable mechanisms; while graphs lacking self-censoring may still be identifiable, they present a fundamentally different model selection problem. Specifically, these candidate graphs are organized into four classes:
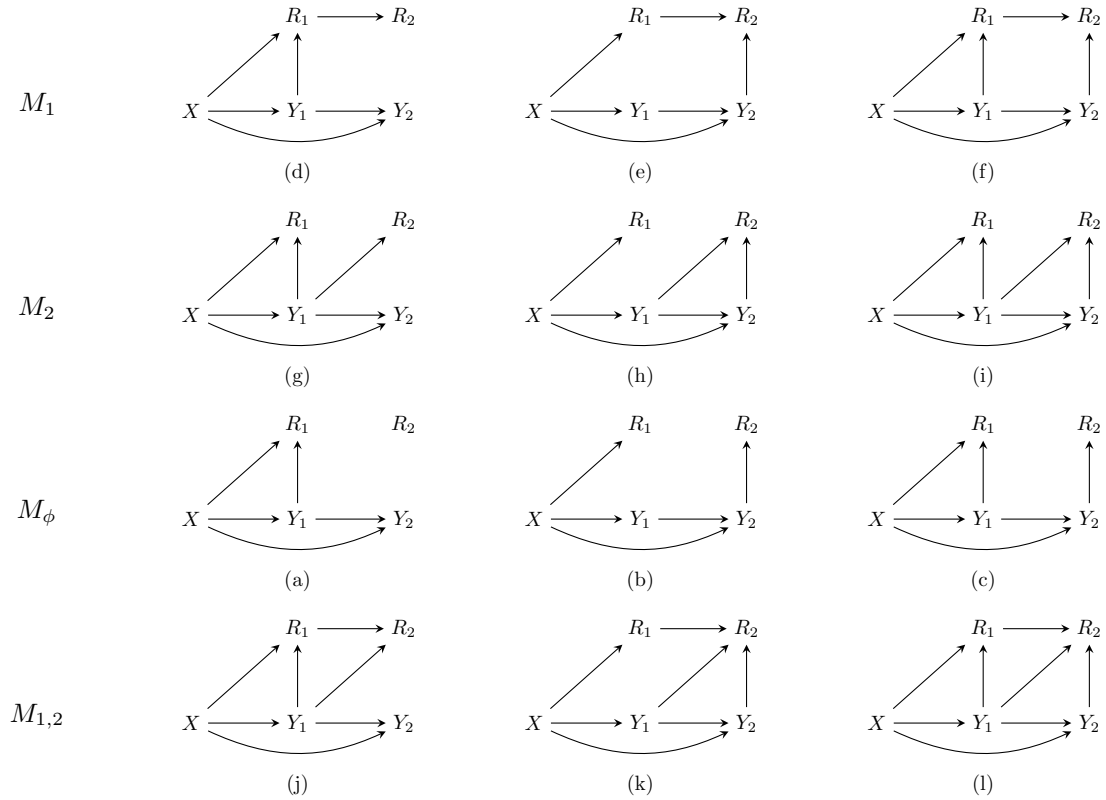
\begin{figure}[ht]
	\begingroup
	\newcommand{\rowlabelmain}[1]{%
		\raisebox{1.05cm}{\parbox{0.15\textwidth}{\centering\bfseries #1}}}
	\newcommand{\repgraphmain}[5]{%
		\begin{tikzpicture}[scale=0.78,transform shape]
			\node (X) at (0,0) {$X$};
			\node (Yone) at (1.8,0) {$Y_1$};
			\node (Ytwo) at (3.6,0) {$Y_2$};
			\node (Rone) at (1.8,1.6) {$R_1$};
			\node (Rtwo) at (3.6,1.6) {$R_2$};
			\draw[-{stealth}] (X) -- (Yone);
			\draw[-{stealth}] (Yone) -- (Ytwo);
			\draw[-{stealth}] (X) to[out=-25,in=-155] (Ytwo);
			\draw[-{stealth}] (X) -- (Rone);
			\ifnum#2=1 \draw[-{stealth}] (Yone) -- (Rone); \fi
			\ifnum#3=1 \draw[-{stealth}] (Ytwo) -- (Rtwo); \fi
			\ifnum#4=1 \draw[-{stealth}] (Rone) -- (Rtwo); \fi
			\ifnum#5=1 \draw[-{stealth}] (Yone) -- (Rtwo); \fi
			\node at (1.8,-1.05) {#1};
	\end{tikzpicture}}
	\centering
	\setlength{\tabcolsep}{0pt}
	\renewcommand{\arraystretch}{1.55}
	\begin{tabular*}{\textwidth}{@{\extracolsep{\fill}}cccc@{}}
		
		\rowlabelmain{$M_1$} &
		\makebox[0.283\textwidth][c]{\repgraphmain{(d)}{1}{0}{1}{0}} &
		\makebox[0.283\textwidth][c]{\repgraphmain{(e)}{0}{1}{1}{0}} &
		\makebox[0.283\textwidth][c]{\repgraphmain{(f)}{1}{1}{1}{0}} \\
		\rowlabelmain{$M_2$} &
		\makebox[0.283\textwidth][c]{\repgraphmain{(g)}{1}{0}{0}{1}} &
		\makebox[0.283\textwidth][c]{\repgraphmain{(h)}{0}{1}{0}{1}} &
		\makebox[0.283\textwidth][c]{\repgraphmain{(i)}{1}{1}{0}{1}} \\
		\rowlabelmain{{\purple $M_{\phi}$}} &
		\makebox[0.283\textwidth][c]{\repgraphmain{(a)}{1}{0}{0}{0}} &
		\makebox[0.283\textwidth][c]{\repgraphmain{(b)}{0}{1}{0}{0}} &
		\makebox[0.283\textwidth][c]{\repgraphmain{(c)}{1}{1}{0}{0}} \\
		\rowlabelmain{{\purple $M_{1,2}$}}  &
		\makebox[0.283\textwidth][c]{\repgraphmain{(j)}{1}{0}{1}{1}} &
		\makebox[0.283\textwidth][c]{\repgraphmain{(k)}{0}{1}{1}{1}} &
		\makebox[0.283\textwidth][c]{\repgraphmain{(l)}{1}{1}{1}{1}}
	\end{tabular*}
	\caption{Candidate missingness graphs within the self-censoring scope of the main analysis. The left entry in each row labels the graph class. Every displayed graph contains at least one of $Y_1\to R_1$ and $Y_2\to R_2$; panels are indexed as (a)--(l).}
	\label{fig:graph_catalogue}
	\endgroup
\end{figure}
\begin{itemize}
	\item {\purple \textbf{$M_1$} (top row): $R_1\to R_2$ is present and $Y_1\to R_2$ is absent. Panel (f) is the full model identified in Theorem~\ref{thm:complete1}.}
	\item {\purple \textbf{$M_2$} (second row): $Y_1\to R_2$ is present and $R_1\to R_2$ is absent. Panel (i) is the full model identified in Theorem~\ref{thm:complete2}.}
	\item {\purple \textbf{$M_{\phi}$} (third row): Neither cross-process arrow is present.}
	\item {\purple \textbf{$M_{1,2}$} (bottom row): Both cross-process arrows are present.}
\end{itemize}
{\purple These four symbols serve as labels to classify the graph structures based on the presence of the two cross-process arrows (rather than denoting set operations). Our primary comparison focuses on $M_1$ and $M_2$; $M_{\phi}$ represents their common reduced submodel, whereas $M_{1,2}$ falls outside the main comparison.} This defined set of graphs establishes the formal scope for our subsequent likelihood-based evaluation. Appendix~\ref{sec:outside_candidate_graphs} examines a representative graph without self-censoring to further clarify this scope restriction.

We now introduce a common restriction shared by the two focal classes, $M_1$ and $M_2$. This assumption reflects the sequential timing of the outcomes and forms the theoretical basis for the identification results developed in the subsequent subsections.

\begin{assumption}\label{assu:indep}
	$R_1 \indep Y_2 \mid (Y_1, X)$.
\end{assumption}

{Assumption~\ref{assu:indep} acts as a structural exclusion restriction motivated by the natural temporal ordering of the outcomes. Graphically, under the standard Markov interpretation, every candidate model in Figure~\ref{fig:graph_catalogue} satisfies this restriction, as conditioning on $(Y_1, X)$ blocks all paths between $R_1$ and $Y_2$. In practice, once the first outcome and baseline covariates are known, the observation status of the first outcome carries no additional predictive power for the second. For example, in the Job Corps application (Section~\ref{sec:app}), $R_1$ indicates whether qualification status is recorded at 30 months, and $Y_2$ denotes fourth-year weekly earnings. The assumption implies that merely responding to the 30-month survey provides no further insight into fourth-year earnings beyond what is already explained by the actual qualification status and baseline characteristics. This restriction can be relaxed to hold given only a subset of covariates. However, we use the full vector $X$ to keep the notation simple.}

\subsection{Identification under Missingness Model $M_1$} 

\label{ssec:identification-H}


	


We first consider the full model $M_1$ in Figure~\ref{fig:graph_catalogue}(f). The following assumption imposes the key restriction needed for identification under this structure.

\begin{assumption}\label{assu:M1_1}
	(i) $X \indep R_2 \mid (R_1,Y_2)$ , (ii) $Y_1 \indep R_2 \mid (X,R_1,Y_2)$ .
	
\end{assumption}

Assumption~\ref{assu:M1_1} formally characterizes the conditional independence structure underlying the sequential missingness model $M_1$  in Figure~\ref{fig:graph_catalogue}(f).
Specifically, condition (i) requires that, conditional on the first-outcome response indicator $R_1$ and the second outcome $Y_2$, the covariates $X$ are independent of the second outcome response indicator $R_2$. This implies that the effect of $X$ on $R_2$ is fully mediated through $(R_1,Y_2)$, which is consistent with the graphical structure in Figure~\ref{fig:graph_catalogue}(f). Condition (ii) requires that, conditional on $(X,R_1,Y_2)$, the first outcome $Y_1$ is independent of $R_2$, thus ruling out a direct dependence of second outcome missingness on the realized first outcome. Nevertheless, as illustrated in Figure~\ref{fig:graph_catalogue}(f), the model still allows $R_2$ to be associated with $R_1$. Consequently, under $M_1$, the two response indicators may be dependent through the earlier response process.

Under Assumptions~\ref{assu:indep} and \ref{assu:M1_1}, the joint distribution $f_{M_1}(X,Y_1,R_1,Y_2,R_2)$ of the complete data under model $M_1$ in Figure~\ref{fig:graph_catalogue}(f) can be expressed as
\begin{align}
	f_{M_1}(X,Y_1,R_1,Y_2,R_2)
	= f(X)f(Y_1,R_1 \mid X)\,
	f(Y_2 \mid Y_1,X)\,
	P(R_2 \mid R_1,Y_2).
	\label{eq:joint-F-M1}
\end{align}
Here $f$ denotes a density or probability mass function, and $P$ denotes a response probability. The factors are the baseline distribution $f(X)$, the joint model $f(Y_1,R_1 \mid X)$, the second outcome model $f(Y_2 \mid Y_1,X)$, and the second outcome response model $P(R_2 \mid R_1,Y_2)$. The covariate distribution is directly observed, but the remaining three factors
cannot be directly identified from the observed data alone. 

The key to recovering the relevant joint-distribution components in \eqref{eq:joint-F-M1} is to identify $P(R_2 \mid R_1,Y_2)$.
To identify $P(R_2 \mid R_1,Y_2)$, we introduce the following response-odds function for each value $r_1\in\{0,1\}$:
\begin{align*} \xi_{1,r_1}(Y_2) = \frac{P(R_2=0 \mid X, R_1=r_1, Y_2)}{P(R_2=1 \mid X, R_1=r_1, Y_2)} = \frac{P(R_2=0 \mid R_1=r_1, Y_2)}{P(R_2=1 \mid R_1=r_1, Y_2)}, \end{align*}
where the simplification on the right-hand side follows from Assumption~\ref{assu:M1_1}. Combining $\xi_{1,r_1}(Y_2)$ with the observed data distribution yields the following equation, for   $r_1\in\{0,1\}$:
\begin{align}\label{eq:inter1} f(X, R_1=r_1, R_2=0) = \int \xi_{1,r_1}(Y_2)  f(X, R_1=r_1, Y_2, R_2=1)  d\nu(Y_2), \end{align}
where $\nu$ is a dominating measure for $Y_2$: the counting measure when $Y_2$ is discrete and the Lebesgue measure when $Y_2$ is continuous. {Completeness makes equation~\eqref{eq:inter1} admit at most one response-odds function in each $R_1$ stratum. A separate condition requiring the two conditional distributions of $Y_2$ across the levels of $Y_1$ to be distinct then recovers the conditional distribution of $Y_1$ given $X$. The following theorem states these sufficient conditions and the resulting identification conclusion.}

\begin{theorem}\label{thm:complete1}
	{Under Assumptions~\ref{assu:indep} and~\ref{assu:M1_1}, suppose that, for each $r_1\in\{0,1\}$, the family $
		\bigl\{f(Y_2\mid X=x,R_1=r_1,R_2=1):x\in\operatorname{supp}(X)\bigr\}$ 
		is complete with respect to $Y_2$. If, for almost every $x$,
		{the conditional distributions of $Y_2$ given $(Y_1=1,X=x)$ and $(Y_1=0,X=x)$ are distinct,}
		then the response-odds functions $\xi_{1,r_1}(Y_2)$ and the full-data distribution $f_{M_1}(X,Y_1,R_1,Y_2,R_2)$ are identified.}
\end{theorem}
	
	The completeness condition in Theorem \ref{thm:complete1}  is central to nonparametric identification, as it ensures the uniqueness of the solution to the integral equation in \eqref{eq:inter1}. It has been widely used to establish identification in inverse problems arising from missing data, causal inference, measurement error, and instrumental variables \citep{newey2003instrumental,d2010new,miao2018identifying,jiang2021identification}.
	In discrete settings, completeness in Theorem \ref{thm:complete1} is equivalent to a rank condition on a conditional probability matrix. For example, if both the covariate $X$ and the outcome variable $Y_2$ are discrete, taking values in $\{x_1,\dots,x_J\}$ and $\{y_1,\dots,y_K\}$ respectively, define the $J\times K$ probability matrix $\Theta_{(x,y)}^{r_1}$ with $(j,k)$-th entry $$P(X=x_j, R_1=r_1, Y_2=y_k, R_2=1).$$ If for each $r_1\in\{0,1\}$, the matrix $\Theta_{(x,y)}^{r_1}$ has rank $K$ and $J \geq K$, then the function $\xi_{1,r_1}(Y_2)$ is uniquely identified. In continuous settings, completeness is a property of the conditional family of $Y_2$ indexed by $X$ within each $(R_1=r_1,R_2=1)$ stratum; {it must be verified for that family and the available support variation in $X$, rather than inferred solely from the choice of a familiar outcome distribution. Proof details of Theorem~\ref{thm:complete1} are given in Appendix~\ref{subsec:detail_M1}.} {The identification results above also justify parametric observed-data likelihoods and plug-in estimation of downstream functionals. The likelihood construction and asymptotic normality of the plug-in functional estimator are provided in Appendix~\ref{sec:likelihood-estimation}, while the details of the EM implementation are given in Appendix~\ref{sec:detail_EM}.
		
		\subsection{Identification under Missingness Model $M_2$}
		
		
		We next consider the full model $M_2$ in Figure~\ref{fig:graph_catalogue}(i). The following assumption imposes the key restriction needed for identification under this structure.
		
		\begin{assumption}\label{assu:M2}
			(i) $R_1 \indep R_2 \mid (X,Y_1)$, (ii) $R_1 \indep R_2 \mid (Y_1,Y_2)$, and (iii) $X \indep R_2 \mid (Y_1,R_1,Y_2)$.
		\end{assumption}
		
		Assumption~\ref{assu:M2} characterizes the full model $M_2$, in which the first outcome may affect second outcome nonresponse, while the two response indicators have no direct edge after conditioning on the relevant substantive variables. Compared with model $M_1$, the key distinction is that the second outcome missingness component depends on $(Y_1,Y_2)$ rather than on $(R_1,Y_2)$. Under Assumptions~\ref{assu:indep} and~\ref{assu:M2}, the full-data distribution factors as
		\begin{align}
			f_{M_2}(X,Y_1,R_1,Y_2,R_2)
			= f(X)f(Y_1,R_1\mid X)
			f(Y_2\mid Y_1,X)
			P(R_2 \mid Y_1,Y_2).
			\label{eq:joint_G_M2}
		\end{align}
		The identification argument is analogous to that for model $M_1$ in Section~\ref{ssec:identification-H}: it again reduces recovery of the joint-distribution components to the identification of response-odds functions through rank and completeness conditions. The difference is that, under model $M_2$, two response-odds functions must be recovered sequentially. Define
		\begin{align*}
			\xi_{1}(x,y_1)
			=\frac{P(R_1=0\mid X=x,Y_1=y_1)}
			{P(R_1=1\mid X=x,Y_1=y_1)},\qquad
			\xi_{2,y_1}(Y_2)
			=\frac{P(R_2=0\mid Y_1=y_1,R_1=1,Y_2)}
			{P(R_2=1\mid Y_1=y_1,R_1=1,Y_2)}.
		\end{align*}
		First, $\xi_{1}(x,y_1)$ is obtained from a finite-dimensional linear system whose coefficient matrix is formed by the observed joint probabilities of $(X,Y_1,R_1,R_2)$ among units with $R_1=1$. Second, conditional on each value of $Y_1$, $\xi_{2,y_1}(Y_2)$ is obtained from the same type of integral equation used for model $M_1$, with the covariates serving to ensure completeness. Since these steps repeat the rank/completeness logic already illustrated above, the full derivation, including the explicit matrix and integral equations, is deferred to Appendix~\ref{subsec:detail_M2}. {For each $X=x$, let
			\[
			\Gamma(x)=
			\left[
			f(X=x,Y_1=y_1,R_1=1,R_2=r_2)
			\right]_{r_2\in\{0,1\},\,y_1\in\{0,1\}}
			\]
			denote the observable coefficient matrix in the first linear system; its rows index $R_2$ and its columns index $Y_1$.} The identification result needed for estimation and model selection is summarized below.
		
		\begin{theorem}\label{thm:complete2}
			{Under Assumptions~\ref{assu:indep} and~\ref{assu:M2}, suppose that $\Gamma(x)$ has rank two for almost every $x$ and that, for each $y_1\in\{0,1\}$, the family
				$
				\bigl\{f(Y_2\mid X=x,Y_1=y_1,R_1=1,R_2=1):x\in\operatorname{supp}(X)\bigr\}$ 
				is complete with respect to $Y_2$. Then $\xi_1(x,y_1)$, $\xi_{2,y_1}(Y_2)$, and the full-data distribution $f_{M_2}(X,Y_1,R_1,Y_2,R_2)$ are identified.}
			
		\end{theorem}

        In Theorem~\ref{thm:complete2}, the two identifying conditions operate in tandem: the rank-two condition yields a unique solution to the bivariate system for the first-outcome response odds, while the completeness condition guarantees a unique solution to the integral equation for the second-outcome response odds given $Y_1$. Jointly, these conditions permit the identification of every component within the full-data factorization~\eqref{eq:joint_G_M2}. Formal proofs and detailed estimation procedures are deferred to Appendices~\ref{subsec:detail_M2} and \ref{sec:likelihood-estimation}, respectively.


        The subsequent exposition focuses on the two full missingness models depicted in Figures~\ref{fig:graph_catalogue}(f) and (i), which serve as the analytical foundation for our proposed likelihood construction, contrast estimation, and formal model selection. The identifiability of their intersection class $M_{\phi}$ naturally follows from these results. Additionally, Appendix~\ref{sec:outside_candidate_graphs} discusses the unidentifiability of Figure~\ref{fig:graph_catalogue}(l) within the $M_{1,2}$ class.
		

		
		\section{Likelihood-Based Missingness Model Selection}
		\label{sec:model_selection}

		\subsection{Motivation for Model Selection}\label{subsec:dev_test}
Within the scope of the candidate models outlined in Figure~\ref{fig:graph_catalogue}, we now face the empirical challenge of choosing between the two identifiable families, $M_1$ and $M_2$. Because these structurally distinct models may produce identical fits if the true mechanism is nested in both, a robust selection procedure must first ensure the data can actually distinguish them. We therefore adopt the two-stage likelihood-ratio framework of \citet{Vuong1989}. The first stage tests whether the observed-data distributions induced by $M_1$ and $M_2$ are observationally distinguishable. If they are not, a scenario consistent with their common submodel {\purple $M_{\phi}$}, the procedure stops without ranking them. Otherwise, the second stage evaluates which model provides a better approximation of the true data-generating process.

		We use the Kullback--Leibler Information Criterion (KLIC) to measure the discrepancy between each candidate model and the true observed-data distribution. For model $M_1$ with 
		observed-data density $f_{M_1}(\cdot\,;\theta^{M_1})$, its KLIC value, equal to the Kullback--Leibler divergence from the true density 
		$f_0(\cdot)$ to $f_{M_1}(\cdot\,;\theta)$ is defined as
		\begin{equation}\label{eq:klic_H}
			\mathrm{KLIC}_{M_1}(\theta^{M_1})
			= E_0\!\left\{\log\frac{f_0(O_i)}{f_{M_1}(O_i;\theta^{M_1})}\right\}
			= E_0\{\log f_0(O_i)\} - E_0\{\log f_{M_1}(O_i;\theta^{M_1})\},
		\end{equation} 
		and the same definition applies to model $M_2$.
		Here, $E_0$ denotes expectation 
		under the true observed-data distribution.  
		Following \citet{Sawa1978} and \citet{Vuong1989}, the better model is the one with the smaller minimum KLIC value.
		Since $E_0\{\log f_0(O_i)\}$ is a constant shared by both expressions, this is equivalent to comparing $\sup_{\theta^{M_1}}\,E_0\{\log f_{M_1}(O_i;\theta^{M_1})\}$ and $\sup_{\theta^{M_2}}\,E_0\{\log f_{M_2}(O_i;\theta^{M_2})\}$. 
		Define the pseudo-true parameter under each candidate model by
		\begin{equation}
			\theta_*^{M_1}
			=
			\arg\max_{\theta^{M_1}}
			E_0\{\log f_{M_1}(O_i;\theta^{M_1})\},
			\quad  ~~\theta_*^{M_2}
			=
			\arg\max_{\theta^{M_2}}
			E_0\{\log f_{M_2}(O_i;\theta^{M_2})\} .
			\label{eq:pseudo_true_parameters}
		\end{equation}
		Model selection therefore reduces to comparing $E_0\{\log f_{M_1}(O_i;\theta_*^{M_1})\}$ and $E_0\{\log f_{M_2}(O_i;\theta_*^{M_2})\}$. 
		These parameters characterize the best approximations to $f_0$ within the two model families $ f_{M_1}$ and $ f_{M_2}$, regardless of whether the models are correctly specified.
		
		Recall that the population comparisons $E_0\{\log f_{M_1}(O_i;\theta_*^{M_1})\}$ and $E_0\{\log f_{M_2}(O_i;\theta_*^{M_2})\}$ 
		depend on the unknown true observed-data distribution $f_0$ and therefore cannot be evaluated directly. Under standard regularity 
		conditions \citep{White1982}, as $n\to\infty$,
		\begin{align}
			\frac{1}{n}\sum_{i=1}^{n}\ell^{M_1}_i(\hat{\theta}^{M_1})
			&\xrightarrow{a.s.} E_0\{\log f_{M_1}(O_i;\theta_*^{M_1})\} \label{eq:as_H}=E_0\{\ell_i^{M_1}(\theta_*^{M_1})\},\\
			\frac{1}{n}\sum_{i=1}^{n}\ell^{M_2}_i(\hat{\theta}^{M_2})
			&\xrightarrow{a.s.} E_0\{\log f_{M_2}(O_i;\theta_*^{M_2})\}=E_0\{\ell_i^{M_2}(\theta_*^{M_2})\}, \label{eq:as_G}
		\end{align}
		where $\ell_i^{M_1}(\hat{\theta}^{M_1})$ and $\ell_i^{M_2}(\hat{\theta}^{M_2})$ denote the individual log-likelihood contributions under models $M_1$ and $M_2$, and $\hat{\theta}^{M_1}$ and $\hat{\theta}^{M_2}$ are the maximum 
		likelihood estimators of $\theta_*^{M_1}$ and $\theta_*^{M_2}$, respectively. 
		Based on \eqref{eq:as_H} and \eqref{eq:as_G}, define
		\[
		LR_n(\theta_*^{M_1},\theta_*^{M_2})
		=\sum_{i=1}^n\{\ell_i^{M_1}(\theta_*^{M_1})-\ell_i^{M_2}(\theta_*^{M_2})\}.
		\]
		The maximized log-likelihood ratio statistic is
		\begin{equation}\label{eq:LR_def}
			LR_n = LR_n(\hat{\theta}^{M_1}, \hat{\theta}^{M_2})
			= \sum_{i=1}^{n} \ell^{M_1}_i(\hat{\theta}^{M_1}) - \sum_{i=1}^{n}\ell^{M_2}_i(\hat{\theta}^{M_2})
			= \sum_{i=1}^{n} \log\frac{f_{M_1}(O_i;\hat{\theta}^{M_1})}{f_{M_2}(O_i;\hat{\theta}^{M_2})}.
		\end{equation}
		It follows that, as $n\to\infty$,
		\begin{equation}\label{eq:lr_consist}
			\frac{1}{n}LR_n \xrightarrow{a.s.} 
			E_0\!\left\{\log\frac{f_{M_1}(O_i;\theta_*^{M_1})}{f_{M_2}(O_i;\theta_*^{M_2})}\right\}=E_0\!\left\{\ell_i^{M_1}(\theta_*^{M_1})-\ell_i^{M_2}(\theta_*^{M_2})\right\}.
		\end{equation}
		A positive limit in \eqref{eq:lr_consist} favors the candidate model $M_1$, while a negative limit favors model $M_2$.
		Thus, testing whether the two models are equally close to the true observed-data distribution is equivalent to testing whether the population limit in \eqref{eq:lr_consist} is zero. 
		The asymptotic distribution of the test statistic under this null depends on whether the variance of the individual log-likelihood ratio,
		\begin{equation}\label{eq:variance}
			\omega_*^2 =  \mathrm{Var}\!\left\{\log\frac{f_{M_1}(O_i;\theta_*^{M_1})}{f_{M_2}(O_i;\theta_*^{M_2})}\right\}=\mathrm{Var}\!\left\{\ell_i^{M_1}(\theta_*^{M_1})-\ell_i^{M_2}(\theta_*^{M_2})\right\},
		\end{equation}
		is positive or zero.
		This distinction motivates the two-stage model selection procedure introduced below, which first examines whether $\omega_*^2=0$ and then conducts model selection accordingly.

        Intuitively, it $\omega_*^2 = 0$ in \eqref{eq:variance} implies that the two best-approximating densities are identical almost everywhere (i.e., $f_{M_1} = f_{M_2}$). Within our framework, this zero-variance scenario occurs naturally when the true data-generating mechanism $f_0$ belongs to the intersection of the two families, {\purple $M_{\phi}$ in Figure~\ref{fig:graph_catalogue}}. This fundamental connection motivates the two-stage procedure introduced below, which first examines whether $\omega_*^2=0$ to rule out observational equivalence before proceeding to rank the models.

		
		\subsection{Model Selection Procedure}\label{subsec:ms_procedure}
		
		Building on the KLIC framework established above, we now present the proposed two-stage model selection procedure for models $M_1$ and $M_2$ in
		Algorithm~\ref{algo:two-stage}. A full theoretical justification is provided in the next subsection.     
		{Let $p$ and $q$ denote the dimensions of the parameter vectors under $M_1$ and $M_2$, respectively.}
		\vspace{2.3em}
        
		\begingroup
		\fontsize{10.5pt}{12.6pt}\selectfont
		\par\noindent\rule{\linewidth}{0.8pt}
		\par\vspace{0.1\baselineskip}
		\refstepcounter{algorithm}
		\label{algo:two-stage}
		\noindent\makebox[\linewidth][c]{\fontsize{13}{13}\selectfont\textbf{Algorithm \thealgorithm:} Two-stage model selection procedure}
		\par\vspace{0.05\baselineskip}
		\par\noindent\rule{\linewidth}{0.4pt}
		\vspace{-0.5\baselineskip}
		\begin{algorithmic}[1]
			\State \textbf{Input:} Observed data $O_1,\ldots,O_n$; models $M_1,M_2$; {significance levels $\alpha_1,\alpha_2$}
			\State \textbf{Preparation:}
			Estimate $\hat\theta^{M_1}$ and $\hat\theta^{M_2}$ from the EM algorithm in Appendix \ref{sec:likelihood-estimation}. For $i=1,\ldots,n$, compute the pointwise log-likelihood differences $m_i = \ell^{M_1}_{i}(\hat{\theta}^{M_1}) - \ell^{M_2}_{i}(\hat{\theta}^{M_2})$, their average $\bar m_n=n^{-1}\sum_{i=1}^n m_i$, the likelihood ratio $LR_n = \sum_{i=1}^n m_i$, and the centered sample variance $\hat\omega_n^2=n^{-1}\sum_{i=1}^n(m_i-\bar m_n)^2$. 
			{Compute the sample Hessian matrices $\hat A_M=n^{-1}\sum_{i=1}^n\nabla_\theta^2\ell_i^M(\hat\theta^M)$ for $M\in\{M_1,M_2\}$, and the score vectors $\hat S_{M_1,i} = \nabla_\theta \ell^{M_1}_{i}(\hat{\theta}^{M_1})$ and $\hat S_{M_2,i} = \nabla_\theta \ell^{M_2}_{i}(\hat{\theta}^{M_2})$.}
			Define
			$\hat B_{M_1} = n^{-1}\sum_{i=1}^n \hat S_{M_1,i}\hat S_{M_1,i}^\T$,   $
			\hat B_{M_2} = n^{-1}\sum_{i=1}^n \hat S_{M_2,i}\hat S_{M_2,i}^\T$, $\hat B_{M_1M_2} = n^{-1}\sum_{i=1}^n \hat S_{M_1,i}\hat S_{M_2,i}^\T,$ 
			and $\hat B_{M_2M_1} = \hat B_{M_1M_2}^\T$. 
			Construct the weight matrix $\hat W \in \mathbb{R}^{(p+q) \times (p+q)}$ defined as:
			\begin{equation*}
				\hat W = \begin{pmatrix}
					-\hat{B}_{M_1} \hat A_{M_1}^{-1} & -\hat{B}_{M_1M_2} \hat A_{M_2}^{-1} \\
					\hat{B}_{M_2M_1} \hat A_{M_1}^{-1} & \hat{B}_{M_2} \hat A_{M_2}^{-1}
				\end{pmatrix},
			\end{equation*}
			and let $\hat\lambda=(\hat\lambda_1,\ldots,\hat\lambda_{p+q})^{\T}$ denote the eigenvalues of $\hat W$.
			
			\State \textbf{Stage 1:} Form the variance test statistic $\hat T_n=n\hat{\omega}_n^2$, and consider the hypotheses:
			\begin{equation}
				\label{null-test-H0-step2}
				\mathcal{H}_0^1 : \omega_*^2 = 0 \quad \text{vs} \quad \mathcal{H}_1^1 : \omega_*^2 > 0.
			\end{equation}
			Let $\hat q_{1-\alpha_1}$ be the $(1-\alpha_1)$ quantile of
			$\sum_{j=1}^{p+q}\hat\lambda_j^2\chi_{1,j}^2$, where the
			$\chi_{1,j}^2$ are independent chi-squared random variables with one
			degree of freedom. If $\hat T_n\leq\hat q_{1-\alpha_1}$, report that the
			data do not distinguish the two models and stop; otherwise proceed to
			Stage~2.
			\State \textbf{Stage 2:} Based on the normalized likelihood ratio
			$\hat V_n =  {LR_n(\hat{\theta}^{M_1}, \hat{\theta}^{M_2})}/({\sqrt{n} \cdot \hat{\omega}_n})$ conduct the model selection test:
			\begin{equation}
				\label{null-test-H0-step3}
				\mathcal{H}_0^2 :  E_0\!\left\{\ell_i^{M_1}(\theta_*^{M_1})-\ell_i^{M_2}(\theta_*^{M_2})\right\}=0, 
			\end{equation}
			and \begin{equation*}
				\quad \mathcal{H}^{2,M_1}_1 :  E_0\!\left\{\ell_i^{M_1}(\theta_*^{M_1})-\ell_i^{M_2}(\theta_*^{M_2})\right\} > 0, \quad \mathcal{H}^{2,M_2}_1 :  E_0\!\left\{\ell_i^{M_1}(\theta_*^{M_1})-\ell_i^{M_2}(\theta_*^{M_2})\right\} < 0.
			\end{equation*}
			{Let $q_2=z_{1-\alpha_2/2}$. If $|\hat V_n|\le q_2$, do not select either model because the data do not provide significant evidence of a difference between their minimum KLIC values. If $\hat V_n>q_2$, select model $M_1$; if $\hat V_n<-q_2$, select model $M_2$.}
			\Statex \vspace{-0.65\baselineskip}
			\Statex \rule{\linewidth}{0.8pt}
		\end{algorithmic}
		\endgroup

Table~\ref{tab:graph_rows_selection} summarizes the expected large-sample behavior of Algorithm~\ref{algo:two-stage} under four possible locations of the true distribution: 
$f_0\in M_{\phi}$, $f_0\in M_1$, $f_0\in M_2$, and $f_0\in M_{1,2}$. 
The procedure compares only the two identifiable candidate models $M_1$ and $M_2$. 
Stage~1 first tests whether their observed-data likelihoods are distinguishable; if not, the procedure stops without ranking the models. 
This situation is expected when $f_0\in M_{\phi}$, where the two models share the same observed-data distribution. 
When $f_0\in M_1$ or $f_0\in M_2$, and the corresponding model has a unique KLIC advantage, Stage~2 consistently favors the candidate model closer to the true distribution. 
When $f_0\in M_{1,2}$, both candidate models are misspecified, and the procedure should be interpreted as comparing their relative observed-data fit rather than identifying the true missingness mechanism. In this case, the procedure may either stop at Stage~1 or fail to find a significant KLIC difference.

		\begin{table}[htbp]
			\centering
\caption{{\wyy True-model configurations in Figure~\ref{fig:graph_catalogue}  and expected results of Algorithm~\ref{algo:two-stage}.}}
			\label{tab:graph_rows_selection}
			\small
			\begin{tabular}{>{\centering\arraybackslash}m{0.20\textwidth} m{0.65\textwidth}}
				\toprule
				{\wyy True model $f_0$ } & Expected result and interpretation\\
				\midrule
				{\purple $f_0\in M_{\phi}$} & Common or reduced cases; Stage~1 tests whether the fitted models are indistinguishable. \\	\addlinespace[1pt]\hline
				\addlinespace[2.5pt]
				$f_0\in M_1$ & Candidate family with $R_1\to R_2$; if distinguishable, Stage~2 may favor $M_1$. \\\addlinespace[1pt]\hline
				\addlinespace[2.5pt]
				$f_0\in M_2$ & Candidate family with $Y_1\to R_2$; if distinguishable, Stage~2 may favor $M_2$. \\\addlinespace[1pt]\hline
				\addlinespace[2.5pt]
				{\purple $f_0\in M_{1,2}$}  & Outside the candidate set. Both candidate models are misspecified. Stage~1 may stop if their
				pseudo-true observed-data distributions are indistinguishable; otherwise,
				Stage~2 may report no significant KLIC difference or favor the candidate
				model that is closer to the true distribution.\\
				\bottomrule
			\end{tabular}
		\end{table}

			\subsection{Theoretical Guarantee for Model Selection}
			\label{subsec:theory}
			We now develop the asymptotic theory for the proposed model selection procedure. Although the two models differ fundamentally in their outcome-missingness specifications, their parameter spaces overlap. Specifically, when the true missingness model reduces to {\purple $M_{\phi}$} in Figure~\ref{fig:graph_catalogue}(c), the same specification can be accommodated by both competing models. Because neither model strictly nests the other and the two are not completely disjoint, they are formally classified as overlapping models, following Definition 3 of \citet{Vuong1989}. To establish the theoretical properties of the proposed model selection procedure, we derive the asymptotic distributions of the test statistics under the regularity conditions introduced in Appendix~\ref{sm_sec:regularity}.
			
			We begin with a preliminary lemma that characterizes the
			null hypothesis $\mathcal{H}_0^{1}:\omega_*^2=0$. This null is equivalent
			to $f_{M_1}(O_i;\theta_*^{M_1})=f_{M_2}(O_i;\theta_*^{M_2})$ almost
			surely under the true observed-data distribution (see Lemma~4.1 of
			\citet{Vuong1989}), as formalized below.
			\begin{lemma}\label{lemma:equivalence_condition}
				Given Assumptions \ref{assu:regularity}(i)--(iii) and (vi) in Appendix~\ref{sm_sec:regularity}, $\omega_{*}^2=0$ if and only if $f_{M_1}(O_i;\theta_*^{M_1})=f_{M_2}(O_i;\theta_*^{M_2})$ almost surely under the true observed-data distribution.
			\end{lemma}
			
			Lemma~\ref{lemma:equivalence_condition} shows that $\omega_*^2 = 0$
serves as a measure of distinguishability between the two models:
when $\omega_*^2 = 0$, the two models are observationally equivalent,
and comparing their relative closeness to the true observed-data
distribution is not meaningful. In particular, this condition includes
the case where both models reduce to the same observed-data model in
$M_{\phi}$ in Figure~\ref{fig:graph_catalogue}. More generally,
$\omega_*^2 = 0$ indicates that the two fitted models are observationally
indistinguishable, regardless of whether the common observed-data
distribution corresponds to a correctly specified or misspecified
scenario. It is therefore essential to test
$\mathcal{H}_0^1:\omega_*^2=0$ before proceeding to the second-stage
comparison of relative fit.
			
			To operationalize this test, we require a consistent estimator of $\omega_*^2$ together with its asymptotic distribution under $\mathcal{H}_0^1$.
			This result provides the foundation for deriving the asymptotic distributions of the test statistics in the subsequent steps.
			\begin{lemma}[Consistency of Variance Estimators]\label{lem:variance_consistency}
				Under Assumption \ref{assu:regularity} in Appendix~\ref{sm_sec:regularity}, the following convergence results hold:
				
				(i) $\hat{\omega}_{n}^{2} \xrightarrow{a.s.} \omega_*^2$;
				
				(ii) $\tilde{\omega}_{n}^{2} \xrightarrow{a.s.} \omega_*^2 + \left(E_0\left\{\ell^{M_1}_{i}(\theta_*^{M_1}) - \ell^{M_2}_{i}(\theta_*^{M_2})\right\}\right)^2$, 
				where $\tilde{\omega}_{n}^{2} \equiv \frac{1}{n}\sum_{i=1}^{n}\left\{\ell^{M_1}_{i}(\hat\theta^{M_1}) - \ell^{M_2}_{i}(\hat\theta^{M_2})\right\}^{2}$. 
			\end{lemma}
			
			Lemma \ref{lem:variance_consistency} establishes the strong consistency of two distinct variance estimators, which differ according to whether the sample mean is subtracted for centering. Specifically, $\hat{\omega}_{n}^{2}$ is the standard sample variance of the differences $m_i = \ell_i^{M_1}(\hat\theta^{M_1}) - \ell^{M_2}_{i}(\hat\theta^{M_2})$, and it consistently estimates the population variance $\omega_*^2$ defined in \eqref{eq:variance}. The statistic $\tilde{\omega}_{n}^{2}$ is an uncentered second-moment estimator that converges to $E_0[\{\ell^{M_1}_{i}(\theta_*^{M_1}) - \ell^{M_2}_{i}(\theta_*^{M_2})\}^2] = \omega_*^2 + \{E_0(\ell^{M_1}_{i}(\theta_*^{M_1}) - \ell^{M_2}_{i}(\theta_*^{M_2})\}^2$ and serves as an intermediate statistic in the asymptotic derivations. For a detailed proof, see Appendix~\ref{pf:lemma1}.
			
			Furthermore, Lemma \ref{lem:variance_consistency} establishes almost sure convergence ($\xrightarrow{a.s.}$),  ensuring that the uncertainty from variance estimation becomes negligible in large samples. We now discuss the asymptotic distribution of the statistic of interest.
			
			\begin{theorem}[Asymptotic Distribution of the Variance Test Statistics]\label{thm:variance_distribution}
				Under Assumptions \ref{assu:regularity}(i)--(vii) in Appendix~\ref{sm_sec:regularity}, if the null hypothesis $\mathcal{H}_0^1 $ in \eqref{null-test-H0-step2} holds, for any $t \geq 0$
				\begin{equation*}
					P(\hat {T}_n \leq t) \rightarrow M_{p+q}(t ;\lambda^2),
				\end{equation*}
				where $\hat{T}_n = n \hat{\omega}_n^2$ and $M_{p+q}(\cdot ;\lambda^2)$  denotes the distribution function of  $\sum_{j=1}^{p+q} \lambda_j^2 \chi_{1,j}^2$, with $\{\chi_{1,j}^2\}_{j=1}^{p+q}$ being a sequence of independent chi-squared random variables, each with one degree of freedom. The weights $\{\lambda_j^2\}_{j=1}^{p+q}$ are the squared eigenvalues of the $(p+q) \times (p+q)$ matrix $W$ defined as:
				\begin{equation*}
					W = \begin{pmatrix}
						-B_{M_1} A_{M_1}^{-1} & -B_{M_1M_2} A_{M_2}^{-1} \\
						B_{M_2M_1} A_{M_1}^{-1} & B_{M_2} A_{M_2}^{-1}
					\end{pmatrix},
				\end{equation*}
				where the sub-matrices are defined in Assumption
				\ref{assu:regularity} in Appendix~\ref{sm_sec:regularity}, including
				$B_{M_2M_1}=B_{M_1M_2}^{\T}$.
			\end{theorem}
			
			{Theorem~\ref{thm:variance_distribution} characterizes the asymptotic null distribution of the variance statistic $\hat{T}_n$. Under $\mathcal{H}_0^1$, $\hat{T}_n$ converges in distribution to a weighted sum of $p+q$ independent $\chi_1^2$ random variables. The eigenvalues $\{\hat\lambda_j\}$ of $\hat W$ consistently estimate $\{\lambda_j\}$, so their squares provide plug-in estimates of the weights $\{\lambda_j^2\}$. The level-$\alpha_1$ test rejects when $\hat T_n>\hat q_{1-\alpha_1}$, as specified in Algorithm~\ref{algo:two-stage}.}
			Based on Lemma \ref{lemma:equivalence_condition}, if $\mathcal{H}_0^1$ is not rejected, the procedure reports insufficient evidence to distinguish the two observed-data likelihoods and terminates without selecting between them.
			If $\mathcal{H}_0^1$ is rejected, the procedure has evidence against observational equivalence and proceeds to the second-stage directional test. Under a fixed alternative with $\omega_*^2>0$, the first-stage rejection probability tends to one, so the models are asymptotically distinguishable and comparison of their minimum KLIC values becomes meaningful.
			
			Next, we establish the asymptotic properties of the test statistic in Stage 2 of Section~\ref{subsec:ms_procedure}. When $\omega_*^2 = 0$, the second-stage test statistic degenerates and its normal approximation fails, which directly motivates the first-stage distinguishability test. The statistic $\hat V_n$ provides a standardized measure of relative model fit. Its asymptotic properties are as follows:
			
			\begin{theorem}[Asymptotic Distribution of Model Selection Test Statistic]\label{thm:lr_normal}
				Under Assumptions \ref{assu:regularity}(i)--(vii) in Appendix~\ref{sm_sec:regularity}, if $\omega_*^2 > 0$ and $\mathcal{H}_0^2$ holds, then, as $n \to \infty$,
				\begin{equation*}
					\frac{LR_n(\theta_*^{M_1}, \theta_*^{M_2})}{\sqrt{n} \omega_*} \xrightarrow{d} N(0,1).
				\end{equation*}
				Moreover, by the consistency of $\hat{\omega}_{n}^{2}$ (Lemma \ref{lem:variance_consistency}), the sample-based statistic satisfies:
				\begin{equation*}
					\hat V_n = \frac{LR_n(\hat{\theta}^{M_1}, \hat{\theta}^{M_2})}{\sqrt{n}\,\hat{\omega}_n} \xrightarrow{d} N(0,1).
				\end{equation*}
			\end{theorem}
			
	Theorem~\ref{thm:lr_normal} establishes the asymptotic normality of the
log-likelihood ratio statistic under the second-stage null. Since the
observed data are independent and identically distributed (Assumption
\ref{assu:regularity}(i) in Appendix~\ref{sm_sec:regularity}), the
pointwise differences
$m_i^*=\ell^{M_1}_{i}(\theta_*^{M_1})-
\ell^{M_2}_{i}(\theta_*^{M_2})$
are also i.i.d. with finite second moment (Assumption
\ref{assu:regularity}(vi) in Appendix~\ref{sm_sec:regularity}). Under
$\mathcal{H}_0^2$, $E_0(m_i^*)=0$, and hence the central limit theorem gives
\[
\frac{LR_n(\theta_*^{M_1},\theta_*^{M_2})}{\sqrt{n}}
=\frac{1}{\sqrt{n}}\sum_{i=1}^n m_i^*
\xrightarrow{d}N(0,\omega_*^2).
\]
Standard normalization yields the first assertion of the theorem. The
condition $\omega_*^2>0$ ensures non-degenerate variability, so the
cumulative log-likelihood ratio admits a Gaussian limit. Replacing
$\omega_*^2$ by the consistent estimator $\hat\omega_n^2$ and the
pseudo-true parameters by their MLEs yields the feasible statistic
$\hat V_n$ by Slutsky's theorem and standard likelihood expansion. The
detailed proof is provided in Appendix~\ref{pf:thm4}.

The asymptotic normality of $\hat V_n$ provides the theoretical
justification for the Stage~2 selection rule in Algorithm~\ref{algo:two-stage}.
Conditional on rejecting $\mathcal{H}_0^1$ in Stage~1, if
$\hat V_n>q_2$, the procedure selects $M_1$, indicating that $M_1$
has a larger expected log-likelihood and a smaller minimum KLIC value.
If $\hat V_n<-q_2$, the procedure selects $M_2$. When
$|\hat V_n|\leq q_2$, neither model is selected because the data do not
provide sufficient evidence that one candidate model has a smaller
minimum KLIC value than the other.

			\section{Inference After Missingness-Model Selection}
			\label{sec:post_selection}
			
			\begingroup

	Algorithm~\ref{algo:two-stage} uses the observed-data likelihood to determine whether the two candidate missingness models are distinguishable and, if so, whether the data favor $M_1$ or $M_2$ in Figure~\ref{fig:graph_catalogue}. This section considers inference after either candidate model is selected. Let $\widehat M$ denote the model selected by Algorithm~\ref{algo:two-stage} whenever the procedure returns one of the two candidates. For concreteness, we present the results for $\widehat M=M_1$; the corresponding results for $M_2$ follow by reversing the likelihood-ratio contrast and relabeling the model-specific quantities. Once selected, $M_1$ in Figure~\ref{fig:graph_catalogue}(f) is used as a working law for recovering the full-data distribution and evaluating the population functional in \eqref{eq:general_functional}. Let $\hat\tau^{M_1}$ be the generic plug-in estimator in \eqref{eq:functional_main} in Appendix~\ref{sec:likelihood-estimation}, and let $\tau_*^{M_1}$ denote its population target. Thus, $\hat\tau^{M_1}$ targets the value obtained by evaluating the $M_1$ full-data law at the pseudo-true parameter $\theta_*^{M_1}$ defined in \eqref{eq:pseudo_true_parameters}; this value equals the true $\tau(F)$ when $M_1$ is correctly specified.
			
		Let $\mu_*=\mathrm{KLIC}_{M_2}(\theta_*^{M_2})-\mathrm{KLIC}_{M_1}(\theta_*^{M_1})$ in \eqref{eq:klic_H}; hence, $\mu_*>0$ indicates that $M_1$ has the smaller KLIC value. 
According to Algorithm~\ref{algo:two-stage}, $M_1$ is selected only when Stage~1 rejects indistinguishability and the Stage~2 statistic exceeds its upper critical value:
\[
\{\widehat M=M_1\}
=
\{\hat T_n>\hat q_{1-\alpha_1},\ \hat V_n>q_2\}.
\]
Under the fixed-law setting considered below, if $\omega_*^2>0$, Stage~1 rejects with probability tending to one. Moreover, if $\mu_*>0$, then $\hat V_n\to_p+\infty$, and the Stage~2 selection criterion $\hat V_n>q_2$ also holds with probability tending to one. The following corollary formalizes this implication.
			
			\begin{corollary}[Consistent Selection of $M_1$]
				\label{cor:selection_determinate}
				
				Suppose Assumption~\ref{assu:regularity} in Appendix~\ref{sm_sec:regularity} holds and the data-generating
				law is fixed with $\omega_*^2>0$.  If $\mu_*>0$, then
				\[
				P(\widehat M=M_1)\longrightarrow 1.
				\]
			\end{corollary}
			
			Corollary~\ref{cor:selection_determinate} states that when $M_1$ has a fixed positive expected likelihood advantage $\mu_*>0$, the two-stage procedure selects $M_1$ with probability tending to one. Conditioning on $\{\widehat M=M_1\}$ is then
			asymptotically equivalent to not conditioning at all, and the plug-in
			estimator retains its ordinary limiting distribution. This is made
			precise below using the asymptotic normality of the plug-in functional
			in Theorem~\ref{thm:ate_normality} in Appendix~\ref{sec:likelihood-estimation}.
			
			\begin{theorem}[Valid fixed-model inference after selection]
				\label{thm:post_selection_wald}
				Suppose the conditions of Theorem~\ref{thm:ate_normality} in Appendix~\ref{sec:likelihood-estimation} hold for
				model $M_1$, and that the data-generating law is fixed with
				$\omega_*^2>0$ and $\mu_*>0$. Then
				\[
				\sqrt n\bigl(\hat\tau^{M_1}-\tau_*^{M_1}\bigr)
				\;\Big|\;\{\widehat M=M_1\}
				\xrightarrow{d}
				N(0,V_{M_1}),
				\]
				where $V_{M_1}$ is defined in \eqref{eq:ate_variance} in Appendix~\ref{sec:likelihood-estimation}.  Consequently, if
				$\hat V_{M_1}\to_p V_{M_1}$, the ordinary Wald interval
				\[
				\mathcal C_{1-\alpha}^{\mathrm{Wald}}
				=
				\left[
				\hat\tau^{M_1}
				-
				z_{1-\alpha/2}\sqrt{\frac{\hat V_{M_1}}{n}},~
				\;
				\hat\tau^{M_1}
				+
				z_{1-\alpha/2}\sqrt{\frac{\hat V_{M_1}}{n}}~
				\right]
				\]
				satisfies
				\[
				P\left\{
				\tau_*^{M_1}\in\mathcal C_{1-\alpha}^{\mathrm{Wald}}
				\,\middle|\,
				\widehat M=M_1
				\right\}
				\longrightarrow
				1-\alpha.
				\]
			\end{theorem}
			
			Theorem~\ref{thm:post_selection_wald} shows that, when $\mu_*>0$, the two-stage procedure behaves like a consistent
			model selector and the usual fixed-model interval for the selected
			working functional is asymptotically valid. The coverage statement
			concerns the model-specific population functional
			$\tau_*^{M_1}$ implied by the selected working model. It does not assert
			coverage of a target that is invariant across all possible selection
			outcomes.
			
			{When $\sqrt n\,\mu_{*,n}\to\delta$, selection remains random at first order and the joint limit in \eqref{eq:joint_post_expansion} in Appendix~\ref{sec:proof_post_selection} yields a selection-normal conditional law. Appendix~\ref{app:selection_normal_boundary} first constructs an oracle interval for a known $\delta$, then forms a least-favourable envelope that accounts for uncertainty about $\delta$, and finally defines an elastic interval that returns to ordinary Wald inference when $\hat V_n$ is sufficiently large. The results are pointwise for the cases $\sqrt n\,\mu_{*,n}\to\delta$ and $\mu_*>0$; they do not claim uniform validity over all intermediate sequences.}
			
			\endgroup
			
			\section{Simulation Studies}\label{sec:sim}
			
			\subsection{Simulation Studies for \texorpdfstring{$M_1$ and $M_2$}{M1 and M2}}\label{sec:sim_set}
			We begin with settings in which the generating mechanism belongs to one of the two candidate families. For binary and continuous second outcomes, we consider $M_1^*$ and $M_2^*$ as the generating mechanisms, sample sizes $n\in\{1000,2000,5000\}$, and 200 Monte Carlo replications. Each dataset is fitted under both candidate models, denoted by working model $M_1$ and  working model $M_2$. We evaluate estimation of $\tau_{\mathrm{SC}}(F)$ in \eqref{eq:standardized_contrast} using Monte Carlo bias, the mean estimated standard error(SE), and coverage probability (CP) of the nominal 95\% confidence interval; Bias and SE are multiplied by 100, and CP is reported as a percentage. Selection accuracy (Acc.) is the proportion of valid Monte Carlo replications
			in which Algorithm~\ref{algo:two-stage} favors the candidate model corresponding 
			to the data-generating mechanism. These definitions are used throughout this section, and computational details are given in Appendix~\ref{sec:detail_EM}.
			
			The two generating mechanisms share the outcome model
			\begin{align*}
				X&\sim N(0,1),\qquad
				Y_1\mid X\sim\operatorname{Bernoulli}\{\operatorname{logit}^{-1}(-0.5+X)\},\\
				\eta(Y_1,X)&=-1+0.8X+1.5Y_1+0.5XY_1,\\
				Y_2\mid Y_1,X&\sim
				\begin{cases}
					\operatorname{Bernoulli}\{\operatorname{logit}^{-1}(\eta(Y_1,X))\},&\text{binary outcome},\\
					N\{\eta(Y_1,X),1\},&\text{continuous outcome},
				\end{cases}
			\end{align*}
			where $\operatorname{logit}^{-1}(u)=\exp(u)/\{1+\exp(u)\}$. They also share
			\[
			P(R_1=1\mid X,Y_1)=\operatorname{logit}^{-1}(-0.8+0.5X+0.3Y_1).
			\]
			The later response mechanism distinguishes the two families:
			\begin{align*}
				M_1^*:\quad
				P(R_2=1\mid R_1,Y_2)
				&=\operatorname{logit}^{-1}(-1+1.2R_1+0.4Y_2),\\
				M_2^*:\quad
				P(R_2=1\mid Y_1,Y_2)
				&=\operatorname{logit}^{-1}(-1+2.5Y_1+0.4Y_2).
			\end{align*}
			These mechanisms correspond to panels (f) and (i), respectively, in 
			Figure~\ref{fig:graph_catalogue}. Table~\ref{tab:sim_ms} summarizes the 
			estimation results under the data-generating mechanisms $M_1^*$ and $M_2^*$. 
			For each data-generating mechanism, we fit both candidate models $M_1$ and 
			$M_2$, allowing us to examine the performance of the estimators under both 
			correct model specification and model misspecification. We also report the 
			performance of the two-stage procedure for selecting between the two candidate 
			models.
			
			\begin{table}[ht]
				\centering
				\caption{Estimation and model-selection performance under $M_1$ and $M_2$.}\label{tab:sim_ms}
				\resizebox{\textwidth}{!}{
					\begin{tabular}{llccccccccccccccc}
						\toprule
						& & \multicolumn{7}{c}{Binary} & & \multicolumn{7}{c}{Continuous} \\
						\cmidrule(lr){3-9} \cmidrule(lr){11-17}
						& & \multicolumn{3}{c}{Working model $M_1$} &\multicolumn{3}{c}{Working model $M_2$} & &
						& \multicolumn{3}{c}{Working model $M_1$} & \multicolumn{3}{c}{Working model $M_2$} &\\
						\cmidrule(lr){3-5} \cmidrule(lr){6-8} \cmidrule(lr){11-13} \cmidrule(lr){14-16}
						& $n$ & Bias & SE & CP & Bias & SE & CP & Acc. & &Bias & SE & CP & Bias & SE & CP & Acc. \\
						\midrule
						\multirow{3}{*}{$M_1^*$}
						& 1000 &  0.58 &  7.14 & 98.00 & $-$7.16 & 6.73 & 78.00 & $96.5\%$ && $-$0.40 & 14.25 & 93.50 &   1.43 & 15.61 & 96.00 & $96.5\%$ \\
						& 2000 &  0.25 &  5.01 & 94.50 & $-$7.07 & 4.76 & 62.00 & $100\%$& &  0.29 & 10.11 & 92.00 &   2.15 & 11.08 & 92.50 & $100\%$ \\
						& 5000 &  0.02 &  3.08 & 93.50 & $-$7.24 & 2.95 & 33.50 & $100\%$ && $-$0.07 &  6.39 & 96.50 &   1.65 & 7.02 & 94.00 & $100\%$ \\
						\midrule
						\multirow{3}{*}{$M_2^*$}
						&  1000 & 12.63 &  4.87 & 25.00 & $-$0.72 & 9.83 & 94.50 & $91.5\%$ && 29.10 & 16.93 & 57.50 & $-$5.63 &  17.75 & 95.00 & $98.5\%$ \\
						& 2000 & 12.71 &  3.43 &  3.50 & $-$0.14 & 6.62 & 93.00 & $100\%$& & 29.71 & 11.94 & 27.00 & $-$5.40 &  12.40 & 96.50 & $100\%$ \\
						& 5000 & 12.64 &  2.15 &  0.00 & $-$0.19 & 4.25 & 95.50 & $100\%$ && 29.36 &  7.49 &  2.00 & $-$4.95 &   7.68 & 91.00 & $100\%$ \\
						\bottomrule
				\end{tabular}}
			\end{table}

			Table~\ref{tab:sim_ms} shows a clear distinction between correct model 
			specification and model misspecification. Under the data-generating mechanism 
			$M_1^*$, the estimator based on the correctly specified working model $M_1$ 
			has scaled biases ranging from $-0.40$ to $0.58$, with coverage probabilities 
			between 92.00\% and 98.00\% for both the binary and continuous outcomes. 
			Similarly, under $M_2^*$, the estimator based on the correctly specified 
			working model $M_2$ exhibits substantially smaller bias than that based on 
			$M_1$, with coverage probabilities generally close to 95\%.
			
			In contrast, fitting the misspecified working model $M_1$ to data generated 
			under $M_2^*$ results in scaled biases of approximately 12.6 for the binary 
			outcome and 29 for the continuous outcome. The corresponding coverage probabilities decline from 
			25.00\% to 0.00\% for the binary outcome and from 57.50\% to 2.00\% for the 
			continuous outcome as $n$ increases. Fitting the misspecified working model 
			$M_2$ under $M_1^*$ is particularly consequential for the binary outcome, for 
			which the coverage probability decreases from 78.00\% to 33.50\%. 
			
			The two-stage selection procedure performs well across all settings. Its 
			selection accuracy ranges from 91.5\% to 98.5\% at $n=1000$ and reaches 100\% 
			in every setting when $n\geq 2000$. These results indicate that the procedure 
			increasingly selects the correctly specified candidate model as the sample 
			size increases, particularly in settings where fitting the wrong missingness 
			model can substantially distort subsequent inference.

			\subsection{Simulation Studies for \texorpdfstring{{\purple $M_{\phi}$}}{M-empty} class}\label{sec:sim_mar}
			We next consider {\purple $M_{\phi}$} to evaluate Stage~1 of Algorithm~\ref{algo:two-stage}. In this setting the two candidate models have the same observed-data density at their population fits, so $\omega_*^2=0$ and the correct decision is to stop at Stage~1 and report indistinguishability. Retaining the continuous outcome model from Section~\ref{sec:sim_set}, we generate the response indicators by
			\begin{align*}
				P(R_1=1\mid X,Y_1)&
				=\text{logit}^{-1}(\gamma_0+\gamma_1X+\gamma_2Y_1)
				=\text{logit}^{-1}(0.8+0.3X+1.2Y_1),\\
				P(R_2=1\mid Y_2)&
				=\text{logit}^{-1}(\delta_0+\delta_3Y_2)
				=\text{logit}^{-1}(0.8+0.3Y_2),
			\end{align*}
			which corresponds to panel (c) of Figure~\ref{fig:graph_catalogue}. We generate 1000 datasets with $n=1000$ and apply the Stage~1 test at nominal levels 0.05, and 0.10. Because rejection is a type~I error here, Figure~\ref{fig:type1} tracks the cumulative rejection rates over the Monte Carlo replications; Theorem~\ref{thm:variance_distribution} predicts that these rates converge to their nominal levels.
			
			\begin{figure}[ht]
				\centering
				\includegraphics[width=0.95\linewidth]{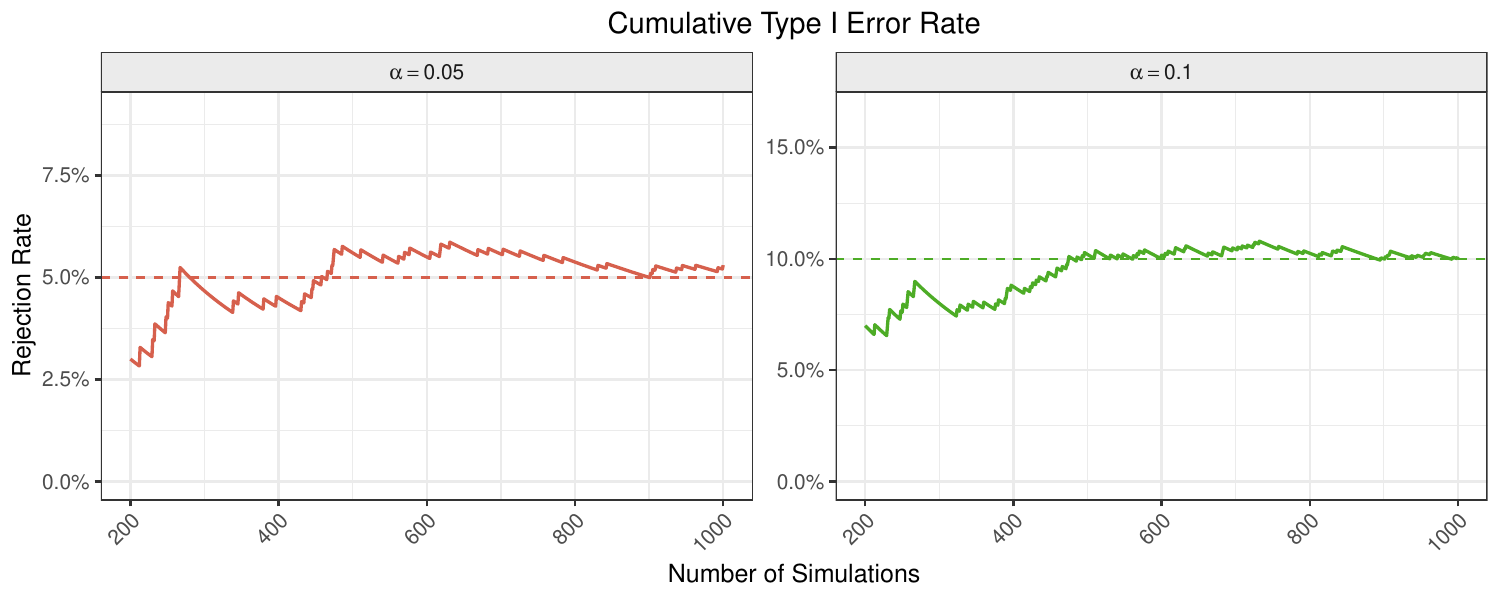}
				\caption{Stage~1 rejection rates at two nominal levels under {\purple $M_{\phi}$}.}
				\label{fig:type1}
			\end{figure}
			
			{\wyy Figure~\ref{fig:type1} shows that, as the number of Monte Carlo replications increases, the cumulative rejection rates gradually stabilize around the corresponding nominal significance levels of $\alpha= 5\%$ and $\alpha= 10\%$, indicating that the Stage~1 test provides satisfactory control of the Type~I error rate. Accordingly, the procedure terminates at Stage~1 without rejecting the null hypothesis of indistinguishability in approximately 95\% and 90\% of the replications, respectively.}
			The close agreement between empirical and nominal rejection rates supports the weighted chi-square calibration in Theorem~\ref{thm:variance_distribution} and confirms that the procedure does not systematically force a Stage~2 ranking when the two fitted observed-data models coincide.
			
			\subsection{Simulation Studies for \texorpdfstring{{\purple $M_{1,2}$}}{M1,2} class}
			Finally, we consider {\purple $M_{1,2}$}, where the generating graph contains both $R_1\to R_2$ and $Y_1\to R_2$ and therefore belongs to neither candidate family. This experiment evaluates the Stage~2 comparison when both fitted models are misspecified: Algorithm~\ref{algo:two-stage} can compare their observed-data likelihoods but cannot recover the generating graph. Retaining the outcome model from Section~\ref{sec:sim_set}, we generate
			\begin{align*}
				P(R_1=1\mid X,Y_1)&
				=\text{logit}^{-1}(-0.8 + 0.5 X +0.3Y_1),\\
				P(R_2=1\mid R_1,Y_1,Y_2)&
				=\text{logit}^{-1}(-1.0 +2.5Y_1+1.2R_1+0.4Y_2),
			\end{align*}
			corresponding to panel (l) of Figure~\ref{fig:graph_catalogue}. {\wyy
				To assess the Monte Carlo stability of this case, we fix $n=1000$ and use 100, 200, and 500 replications. 
				Table~\ref{tab:union_sim} reports the bias(Bias), the mean estimated standard error(SE), and coverage probability (CP) of the nominal 95\% confidence interval for the estimates obtained by fitting both candidate working  models {$M_1$} and {$M_2$} to the simulated data, with Bias and SE scaled by 100.
				It also reports the empirical Stage~2 non-rejection frequency in the last row, corresponding to $|\hat V_n|\le q_2$. 
			}

			\begin{table}[t]
				\centering
				\caption{Functional estimation and Stage~2 no-preference rates under {\purple $M_{1,2}$}.}
				\label{tab:union_sim}
				\small
				\begin{tabular}{llccccccc}
					\toprule
					& & \multicolumn{3}{c}{Working model $M_1$} & \multicolumn{3}{c}{Working model $M_2$} &  \\
					\cmidrule(lr){3-5}\cmidrule(lr){6-8}
					$Y_2$ & \shortstack{Monte Carlo\\replications} & Bias & SE & CP & Bias & SE & CP & \shortstack{Stage-2\\no preference} \\
					\midrule
					\multirow{3}{*}{Binary}
					&100 & 8.15 & 5.28 & 67.0\% & 0.85 & 7.07 & 92.0\% & 95.96\%\\
					&200 & 8.45 & 5.21 & 61.5\% & 0.97 & 7.12 & 93.0\% & 93.97\%\\
					&500 & 8.47 & 5.23 & 63.6\% & 0.54 & 7.14 & 92.4\% & 92.99\%\\
					\midrule
					\multirow{3}{*}{Continuous}
					&100 &13.27&13.64&87.0\%&-2.26&15.15&94.0\%&94.00\%\\
					&200 &14.35&13.54&82.5\%&-0.65&15.09&94.0\%&93.00\%\\
					&500 &15.73&13.57&79.8\%& 0.18&15.07&94.8\%&95.00\%\\
					\bottomrule
				\end{tabular}
			\end{table}
			
			Table~\ref{tab:union_sim} shows that the Monte Carlo results are generally 
			stable as the number of replications increases. For the binary outcome, the 
			scaled bias of the estimator based on working model $M_1$ remains between 
			8.15 and 8.47, with coverage probabilities ranging from 61.5\% to 67.0\%. 
			In comparison, the estimator based on working model $M_2$ has a scaled bias 
			below 1 and coverage probabilities between 92.0\% and 93.0\%. For the 
			continuous outcome, the scaled bias under $M_1$ ranges from 13.27 to 15.73, 
			with coverage probabilities between 79.8\% and 87.0\%, whereas the estimator 
			under $M_2$ has a much smaller scaled bias and coverage probabilities between 
			94.0\% and 94.8\%.
			
			Despite these clear differences in the estimation of the target functional, 
			the Stage~2 statistic falls in the no-preference region in approximately 
			93--96\% of the replications. Because the data-generating mechanism 
			$M_{1,2}$ belongs to neither candidate family, this frequency should not be 
			interpreted as a correct-selection rate. Rather, it indicates that, at the 
			chosen threshold, the observed data generally provide insufficient evidence 
			for Stage~2 to favor one candidate model over the other. These results also 
			illustrate that, when both candidate models are misspecified, their comparison 
			in terms of observed-data KLIC need not reflect their relative performance in 
			estimating a full-data functional. In particular, although working model 
			$M_2$ yields substantially better bias and coverage in this experiment, the 
			Stage~2 procedure frequently expresses no preference between $M_1$ and $M_2$.
			

			\section{Empirical Application}\label{sec:app}
			
			We illustrate the proposed model selection procedure using the Job Corps dataset from the National Job Corps Study (NJCS), a large-scale nationwide evaluation carried out in the mid-1990s across all Job Corps centers in the United States \citep{burghardt2001does}. From approximately 80000 eligible applicants, over 15000 were selected through stratified sampling and randomly assigned to the treatment or control group. The final analysis sample consists of 14125 young adults who completed follow-up interviews through 48 months after randomization.
			
			In this application, subjects were randomly assigned to either the experimental group ($T=1$), where they were offered the opportunity to enroll in the Job Corps program, or the control group ($T=0$). The intermediate outcome $Y_1$ indicates whether an educational or vocational qualification (e.g., a GED or vocational certificate) was obtained by the 30-month follow-up. The final outcome $Y_2$ is weekly earnings measured in the fourth year after randomization. Thirteen baseline covariates are included, covering demographic characteristics (including age, gender, and race/ethnicity), educational attainment, baseline earnings, parental status, and prior arrest history.
			
			{Both outcomes are measured through the multi-wave follow-up, so their observation depends on continued survey participation and may involve related nonresponse processes. There are 1285 individuals missing $Y_1$ only, 1563 missing $Y_2$ only, and 1666 missing both outcomes; thus the total numbers missing $Y_1$ and $Y_2$ are 2951 and 3229, respectively.} Prior analyses of the Job Corps data assumed truncation by death \citep{zhang2009likelihood}, MNAR \citep{huber2012identification}, or conditional MAR \citep{2019Qin}, motivating a re-examination under a nonignorable sequential missingness framework with principled model selection.

			
Table~\ref{tab:real_data_estimation} reports the estimates of the target 
functionals and the maximized observed-data log-likelihood values under the 
two candidate missingness models $M_1$ and $M_2$ in Figure~\ref{fig:graph_catalogue}. 
The columns ``Est.'', ``Var.'', and ``95\% CI'' denote the point estimates, 
estimated variances, and 95\% confidence intervals of the corresponding 
functionals, respectively, while ``LogLik'' denotes the maximized observed-data 
log-likelihood. The maximized observed-data log-likelihood is larger under 
$M_1$ than under $M_2$, suggesting a better fit of $M_1$ among the two 
candidate models. In addition, the two-stage model selection procedure 
described in Section~\ref{subsec:ms_procedure} also favors the missingness model $M_1$. Therefore, 
we use $M_1$ as the selected missingness model for subsequent functional 
interpretation.  

            \begin{table}[htbp]
				\centering
		\caption{Model comparison and functional estimation for the Job Corps data under the two sequential missingness models.}
				\label{tab:real_data_estimation}
				\small
				\begin{tabular}{lcccccccc}
					\toprule
					& \multicolumn{4}{c}{Working model $M_1$} & \multicolumn{4}{c}{Working model $M_2$} \\
					\cmidrule(lr){2-5} \cmidrule(lr){6-9}
					Effect & Est. & Var. & 95\% CI & LogLik& Est. & Var. & 95\% CI & LogLik \\
					\midrule
					{$\tau_{\mathrm{NIE}}$} &  $9.25$ & $0.75$ & $[7.56,\,10.95]$&\multirow{3}{*}{{$-92445.24$}} & $7.81$ & $0.67$ & $[6.21,\,9.41]$  &\multirow{3}{*}{{$-93287.55$}} \\
					{$\tau_{\mathrm{NDE}}$} & $6.28$ & $13.38$ & $[-0.88,\,13.45]$& & $7.00$
					& $13.25$ & $[-0.13,\,14.14]$ &\\
					$\hat\tau_{\mathrm{SC}}$ &  $54.52$& $19.51$ & $[45.86,\,63.17]$&  & $53.48$& $24.42$ & $[43.80,\,63.17]$& \\
					\bottomrule
				\end{tabular}
			\end{table}
            
Table~\ref{tab:real_data_estimation} also reports the natural indirect effect (NIE) and natural direct effect (NDE) under the selected model $M_1$ \citep{pearl2014interpretation,Zuo2025}. Under the selected missingness mechanism, we recover the full-data distribution and estimate the causal effects of the intermediate qualification outcome on fourth-year weekly earnings. The detailed definitions and computational procedures for these quantities are provided in Appendix~\ref{subsec:mediation_functionals}. The estimated NIE is 9.25 (Var. = 0.75; 95\% CI: [7.56, 10.95]), indicating that the indirect effect through the qualification outcome is positive and contributes substantially to the causal effect on fourth-year weekly earnings. The estimated NDE is 6.28 (Var. = 13.38; 95\% CI: [-0.88, 13.45]), representing the direct effect beyond the qualification pathway, although this effect is estimated with greater uncertainty. The plug-in estimate of $\tau_{\mathrm{SC}}(F)$ in \eqref{eq:standardized_contrast} is 54.52 (Var. = 19.51; 95\% CI: [45.86, 63.17]), indicating a substantial causal contrast in fourth-year weekly earnings associated with qualification status under the selected   model. 

			\section{Discussion}
			\label{sec:discussion}
			
			This paper develops a framework for comparing prespecified missingness models when a first outcome and a second outcome are measured sequentially and may be missing not at random. We organize the relevant graph classes, establish identification of full-data distribution components under rank or completeness conditions, and propose a two-stage Vuong-type procedure that first tests observed-data distinguishability and then compares Kullback--Leibler divergences. Substantive summaries are treated as functionals of the recovered distribution rather than as separate primitive targets. We also study inference after a missingness model has been selected: the joint expansion of the functional estimator and the likelihood-ratio statistic leads to selective intervals that account for the selection event. Simulation studies and the Job Corps application illustrate how identification, model comparison, and post-selection inference work together in practice.

            An important point is that the proposed procedure compares candidate
missingness models through their induced observed-data distributions. It
does not, in general, establish that the selected graph is the unique
true missingness mechanism, because distinct latent missingness
mechanisms may lead to the same observed-data distribution. Rather, the
procedure provides a principled likelihood-based comparison among
prespecified candidate models. In particular, when the true mechanism
lies outside the candidate set, the selected model should be interpreted
as the candidate model with better observed-data fit rather than as the
true underlying missingness structure.

Several directions for future research remain. First, extending the proposed framework to more general longitudinal or time-series settings with repeated outcomes and response indicators is an important direction. Second, relaxing the assumption that some covariates do not directly affect the missingness indicators would require new identification results and estimation methods. These extensions are beyond the scope of this paper and are left for future work.
\section*{Data Availability} 
The data are publicly available in the
\href{https://github.com/YingyingWang2026/MSM-SON}{GitHub repository}.

\section*{Supplementary Material}
In the Appendix, we provide detailed identification steps and theorem proofs for the two primary nonignorable missingness models discussed in the main text and present other nonidentified and partially identified missingness mechanisms under similar settings. We then provide proofs for the asymptotic distributions of the test statistics used in the proposed model selection procedure, followed by additional details for the data application.

		\section*{Acknowledgments} This work was supported by the National Natural Science Foundation of China (12401378), the Young Elite Scientists Sponsorship Program of the Beijing High Innovation Plan (NO.20250864), the Beijing Key Laboratory of Applied Statistics and Digital Regulation, the BTBU Digital Business Platform Project by BMEC, the BTBU Research Foundation for Youth Scholars (BRFYS2025) and Stichting Kinderen Kankervrij (KiKa) under the research project 458. Dr. Shanshan Luo would like to thank the Isaac Newton Institute for Mathematical Sciences, Cambridge, for support and hospitality during the programme Causal inference: From theory to practice and back again, where work on this paper was undertaken. This work was supported by EPSRC grant EP/Z000580/1.

 \bibliography{mybib}
 \newpage
\begin{appendix}

\begin{center}
    \bf \huge Appendix
\end{center}
In the Appendix, we provide detailed identification steps and theorem proofs for the two primary nonignorable missingness models discussed in the main text and present other nonidentified and partially identified missingness mechanisms under similar settings. We then provide proofs for the asymptotic distributions of the test statistics used in the proposed model selection procedure, followed by additional details for the data application.

\section{Details of Identification }\label{sec:detail_inden}
			
			This section supplies the identification details for the two candidate
			families $M_1$ and $M_2$ used in the main likelihood comparison. A
			separate scope discussion of graphs outside this candidate set is given
			in Appendix~\ref{sec:outside_candidate_graphs}.
			
			{Throughout this section, we use the standard support regularity for missing-data identification: every response probability that appears in a denominator is positive on the relevant support, and the observed cells used in the displayed linear systems have positive probability. These conditions justify the odds ratios and Bayes inversions below; they are kept here because they are technical rather than substantive restrictions on the candidate graphs.}
			
			\begin{lemma}[Sufficient uniqueness condition via completeness]
				\label{lem:unique_id}
				Let $\nu$ denote the Lebesgue measure if $Y_2$ is continuous and 
				the counting measure if $Y_2$ is discrete. Suppose the integral 
				equation
				\begin{align}
					\int \xi(Y_2)\, f(Y_2 \mid Z)\, d\nu(Y_2) = h(Z) 
					\quad \text{for almost every } Z,
					\label{eq:generic_integral}
				\end{align}
				where $h(Z)$ and $f(Y_2\mid Z)$ are identifiable from the observed 
				data. Then \eqref{eq:generic_integral} admits at most one integrable solution 
				$\xi(Y_2)$ on the support of $Y_2$ if the conditional 
				distribution $f(Y_2\mid Z)$ is complete, i.e., for any measurable 
				function $g$ with $E|g(Y_2)|<\infty$,
				\begin{align}
					\int g(Y_2)\, f(Y_2\mid Z)\, d\nu(Y_2) = 0 
					\quad \text{for almost every }Z
					\implies g(Y_2) = 0
					\quad\text{almost surely}
				\end{align}
				In the discrete case, uniqueness of an unrestricted solution to the
				corresponding linear system is equivalent to the full-column-rank
				condition $\mathrm{Rank}(\Phi)=K$, where $K$ is the number of
				categories of $Y_2$. Completeness is used here as a sufficient
				identification condition for the response-odds functions.
			\end{lemma}
			
			\begin{proof}
				\emph{(Sufficiency.)} Suppose $\xi^{(1)}(Y_2)$ and $\xi^{(2)}(Y_2)$ 
				both satisfy \eqref{eq:generic_integral}. Their difference 
				$g(Y_2) = \xi^{(1)}(Y_2) - \xi^{(2)}(Y_2)$ satisfies
				\begin{align*}
					\int g(Y_2)\, f(Y_2\mid Z)\, d\nu(Y_2) = 0 
					\quad \text{for almost every }Z.
				\end{align*}
				By completeness, $g(Y_2) = 0$ almost surely, so 
				$\xi^{(1)}(Y_2) = \xi^{(2)}(Y_2)$ on the relevant support.
				
				\emph{(Discrete case.)} When $Y_2$ takes $K$ values and $Z$ 
				takes $J$ values, \eqref{eq:generic_integral} is the linear 
				system $\Phi\,\xi = h$ with $\Phi\in\mathbb{R}^{J\times K}$.
				The null space of $\Phi$ is trivial exactly when
				$\mathrm{Rank}(\Phi)=K$, which proves uniqueness of the unrestricted
				linear-system solution.
			\end{proof}
			
			\subsection{Identification of Model $M_1$}\label{subsec:detail_M1}

			Following the decomposition  \eqref{eq:joint-F-M1}, identification proceeds in four steps.
			
			\textbf{Step 1}: Identify $P(R_2 \mid  R_1, Y_2)$.
			
			Under Assumption~\ref{assu:M1_1}, {define, for $r_1\in\{0,1\}$,}
			\begin{align*}
				{\xi_{1,r_1}(Y_2)=\frac{P(R_2=0\mid X,R_1=r_1,Y_2)}{P(R_2=1\mid X,R_1=r_1,Y_2)}=\frac{P(R_2=0\mid R_1=r_1,Y_2)}{P(R_2=1\mid R_1=r_1,Y_2)}},
			\end{align*}
			where the second equality follows from Assumption~\ref{assu:M1_1}. 
			Let $\nu$ denote the Lebesgue or counting measure according 
			to whether $Y_2$ is continuous or discrete.
			Then {for each $r_1\in\{0,1\}$} the following integral equation holds
			\begin{align*}
				{\cfrac{f(X,R_1=r_1,R_2=0)}{f(X,R_1=r_1,R_2=1)}=\int \xi_{1,r_1}(Y_2)f(Y_2 \mid X,R_1=r_1,R_2=1)d\nu(Y_2)},
			\end{align*}
			where the ratio on the left and the integral kernel on the right are identifiable from the observed data. {Lemma~\ref{lem:unique_id} and the stratum-specific completeness condition in Theorem~\ref{thm:complete1} imply that each $\xi_{1,r_1}(Y_2)$ is unique. Positivity then recovers $P(R_2=1\mid R_1=r_1,Y_2)=\{1+\xi_{1,r_1}(Y_2)\}^{-1}$.}

			
			
			
			
			
			\textbf{Step 2:} Identify $ f(Y_2 \mid X, Y_1) $.
			
			By $ R_1 \perp\!\!\!\perp Y_2 \mid (X, Y_1) $ (Assumption \ref{assu:indep}) and Bayes' rule,
			$$f(Y_2 \mid X, Y_1) =f(Y_2 \mid X,Y_1,R_1=1) = \frac{f(Y_2,R_2=1 \mid X,Y_1,R_1=1)}{P(R_2=1 \mid X,Y_1,R_1=1,Y_2)},$$
			where the numerator $ f(Y_2, R_2 = 1 \mid X, Y_1, R_1 = 1) $ can be identified from the observed data. By Assumption~\ref{assu:M1_1} , $ P(R_2 = 1 \mid X, Y_1, R_1 = 1, Y_2)= P(R_2 = 1 \mid  R_1 = 1, Y_2)$, which is identified in Step~1. Hence, $ f(Y_2 \mid X, Y_1) $ is identified.

			\textbf{Step 3:} Identify $f(Y_2\mid X)$.
			
			
			
			By the law of total probability,
			$$f(Y_2\mid X) = \sum_{r_1=0}^{1} f(Y_2, R_1=r_1 \mid X) = \sum_{r_1=0}^{1} \frac{f(Y_2,R_2=1,R_1=r_1\mid X)}{P(R_2=1 \mid X,R_1=r_1,Y_2)},$$
			{where the numerator $f(Y_2,R_2=1,R_1=r_1\mid X)$ is observed and the denominator is identified in Step~1. Hence $f(Y_2\mid X)$ is identified.}

			\textbf{Step 4:} Identify $f(Y_1\mid X)$ and $f(Y_1,R_1\mid X)$.
			
			As $Y_1$ is binary, by the law of total probability, we have
			\begin{align*}
				f(Y_2 \mid X) = f(Y_2 \mid Y_1=1,X) P(Y_1=1 \mid X) + f(Y_2 \mid Y_1=0,X) P(Y_1=0 \mid X),
			\end{align*} 
			where $f(Y_2\mid X)$ and $f(Y_2 \mid X,Y_1) $ are identified in Step~3 and 2, respectively.
			{By the mixture-separation condition in Theorem~\ref{thm:complete1}, the equation has a unique solution $P(Y_1=1\mid X)$ for almost every $X$, from which $P(Y_1\mid X)$ is identified.}
			Since $P(Y_1,R_1=1\mid X)$ is observed, the identity
			$P(Y_1,R_1=0\mid X) = P(Y_1\mid X) - P(Y_1,R_1=1\mid X)$
			identifies $f(Y_1,R_1\mid X)$, completing identification 
			of~\eqref{eq:joint-F-M1}.
			\subsection{Identification of Model $M_2$}\label{subsec:detail_M2}

			\textbf{Step 1}: Identify $P(R_1\mid Y_1,X)$ 
			and $P(Y_1,R_1\mid X)$.
			
			{By $R_1 \perp\!\!\!\perp R_2 \mid (X,Y_1)$ in Assumption~\ref{assu:M2}(i), define
				$$
				\xi_1(x,y_1)
				=\frac{P(R_1=0\mid X=x,Y_1=y_1,R_2=r_2)}
				{P(R_1=1\mid X=x,Y_1=y_1,R_2=r_2)}
				=\frac{P(R_1=0\mid X=x,Y_1=y_1)}
				{P(R_1=1\mid X=x,Y_1=y_1)},
				$$
				which does not depend on $r_2$. For each fixed $x$ and $r_2\in\{0,1\}$,
				\begin{align}\label{eq:inter2}
					f(R_1=0,R_2=r_2,X=x)
					=\sum_{y_1=0}^{1}f(X=x,Y_1=y_1,R_1=1,R_2=r_2)\,\xi_1(x,y_1).
				\end{align}
				Writing the two equations for $r_2=0,1$ in matrix form gives
				$$
				\begin{pmatrix}
					f(R_1=0,R_2=0,X=x)\\
					f(R_1=0,R_2=1,X=x)
				\end{pmatrix}
				=\Gamma(x)
				\begin{pmatrix}\xi_1(x,0)\\ \xi_1(x,1)\end{pmatrix}.
				$$
				The rank condition in Theorem~\ref{thm:complete2} therefore identifies both response odds at almost every $x$, and hence identifies $P(R_1\mid X,Y_1)$. Bayes' rule then gives}
			$$
			P(Y_1 \mid X)
			= \frac{P(Y_1, R_1 = 1 \mid X)}{P(R_1 = 1 \mid X, Y_1)},
			$$
			where the numerator is observed, and the denominator is identified above. 
			Therefore, $P(Y_1 \mid X)$ is identifiable. Using
			$$
			P(Y_1, R_1 = 0 \mid X)
			= P(Y_1 \mid X) - P(Y_1, R_1 = 1 \mid X)
			$$
			identifies $P(Y_1, R_1 \mid X)$.

			\textbf{Step 2}: Identify $P(R_2 \mid  Y_1, Y_2)$.
			
			{Assumptions~\ref{assu:M2}(ii)--(iii) imply that, for each $y_1\in\{0,1\}$,
				\[
				\xi_{2,y_1}(Y_2)
				= \frac{P(R_2=0\mid X,Y_1=y_1,R_1=1,Y_2)}
				{P(R_2=1\mid X,Y_1=y_1,R_1=1,Y_2)}
				= \frac{P(R_2=0\mid Y_1=y_1,Y_2)}
				{P(R_2=1\mid Y_1=y_1,Y_2)}.
				\]}
			
			By the law of total probability, we have
			\begin{align}\label{eq:inter3}
				f(R_2 = 0, X, Y_1 = y_1, R_1 = 1)
				= {\int f(R_2 = 1, X, Y_1 = y_1, R_1 = 1, Y_2)\cdot \xi_{2,y_1}(Y_2)\, d\nu(Y_2)},
			\end{align}
			and thus
			\begin{align*}
				\cfrac{f(R_2 = 0, X, Y_1 = y_1, R_1 = 1)}{f(R_2 = 1, X, Y_1 = y_1, R_1 = 1)}
				= {\int f(Y_2 \mid R_2 = 1, X, Y_1 = y_1, R_1 = 1)\cdot \xi_{2,y_1}(Y_2)\, d\nu(Y_2)},
			\end{align*}
			where the left-hand side and the kernel function are identifiable from the observed data. 
			{By Lemma~\ref{lem:unique_id} and the completeness condition in Theorem~\ref{thm:complete2}, equation~\eqref{eq:inter3} has a unique solution separately for each $y_1$. Thus $\xi_{2,y_1}(Y_2)$ and $P(R_2\mid Y_1,Y_2)$ are identified.}

			\textbf{Step 3:} Identify $f(Y_2 \mid Y_1, X). $
			
			{Under Assumption~\ref{assu:indep} and by Bayes' rule,}
			
			$$f(Y_2 \mid  Y_1, X)=
			f(Y_2 \mid R_1=1, Y_1, X) = \frac{f(Y_2, R_2=1 \mid Y_1, R_1=1, X)}{P(R_2=1 \mid Y_1, R_1=1, Y_2, X)},
			$$
			where the numerator is observed, and the denominator is identified in Step 2. Hence, $ f(Y_2 \mid Y_1, X) $ is identifiable.
			
			Since $f(X)$ is directly observable, each factor of the joint distribution \eqref{eq:joint_G_M2} is identifiable, completing nonparametric identification of the joint distribution.

				\section{Likelihood-Based Estimation}
				\label{sec:likelihood-estimation}
				
				The identification results above justify likelihood-based estimation under either missingness model. In practice, we specify parametric models for $f(Y_1,R_1\mid X)$, $f(Y_2\mid Y_1,X)$, and the second-outcome response probability, which is $P(R_2\mid R_1,Y_2)$ under model $M_1$ and $P(R_2\mid Y_1,Y_2)$ under model $M_2$. Let $\theta^{M_1}$ and $\theta^{M_2}$ collect the corresponding parameters. For $M\in\{M_1,M_2\}$, the observed-data likelihood is obtained by integrating the complete-data likelihood over the missing components:
				\begin{align}
					L_{\mathrm{obs}}(\theta^M)
					=\prod_{i=1}^n \int\!\!\int
					f_M(Y_{1i},Y_{2i},R_{1i},R_{2i}\mid X_i;\theta^M)
					dY_{1i}^{\mathrm{mis}}dY_{2i}^{\mathrm{mis}},
					\label{eq:obs_lik_main}
				\end{align}
				where the integral is degenerate for variables that are observed. The
				maximum likelihood estimator is
				\begin{align}
					\hat{\theta}^M=\arg\max_{\theta^M}\ell_{\mathrm{obs}}(\theta^M),
					\qquad
					\ell_{\mathrm{obs}}(\theta^M)=\log L_{\mathrm{obs}}(\theta^M),
					\label{eq:mle_main}
				\end{align}
				which is computed using the EM procedures in
				Appendix~\ref{sec:detail_EM}.
				
				For a functional of the form \eqref{eq:general_functional}, let
				$h_\theta(x)$ denote $h_F(x)$ evaluated at the full-data law indexed by
				the working-model parameter $\theta$. Its plug-in estimator and
				model-specific population target are
				\begin{align}
					\hat{\tau}^M
					&=\frac1n\sum_{i=1}^n h_{\hat\theta^M}(X_i),
					\label{eq:functional_main}\\
					\tau_*^M
					&=E_0\{h_{\theta_*^M}(X)\},
					\label{eq:functional_target}
				\end{align}
				where $\theta_*^M$ is the pseudo-true parameter under working model $M$.
				For the standardized contrast $\tau_{\mathrm{SC}}(F)$ already defined in
				\eqref{eq:standardized_contrast}, the corresponding plug-in estimator is
				\begin{align}
					\hat{\tau}_{\mathrm{SC}}^M
					=\frac{1}{n}\sum_{i=1}^n
					\left[
					E\{Y_2\mid Y_1=1,X_i;\hat{\theta}^M\}
					-E\{Y_2\mid Y_1=0,X_i;\hat{\theta}^M\}
					\right].
					\label{eq:g_formula_main}
				\end{align}
				Equation~\eqref{eq:g_formula_main} is the model-based plug-in estimator
				associated with the population parameter in
				\eqref{eq:standardized_contrast}. Under correct specification of $M$, its
				probability limit is $\tau_{\mathrm{SC}}(F)$. Under misspecification, the generic target
				$\tau_*^M$ in \eqref{eq:functional_target} records the corresponding
				model-specific probability limit. This plug-in estimator is used in both
				the simulation studies and the empirical application.
				
				For the variance calculation, let
				\[
				d_M=E_0\{\nabla_\theta h_{\theta_*^M}(X)\},
				\qquad
				s_i^M=\nabla_\theta\log f_M(O_i;\theta_*^M),
				\]
				and define the influence function
				\begin{equation}
					\psi_i^M
					=h_{\theta_*^M}(X_i)-\tau_*^M
					-d_M^\T A_M^{-1}s_i^M.
					\label{eq:functional_if}
				\end{equation}
				
				The following theorem establishes the asymptotic normality of the plug-in functional estimator.
				\begin{theorem}[Asymptotic Normality of the Plug-in Functional]
					\label{thm:ate_normality}
					Let $M \in \{M_1, M_2\}$ be a working model. Suppose the regularity conditions for the observed-data MLE in Theorem 5.41 of \citet{van2000asymptotic} hold, $h_\theta(x)$ is continuously differentiable in a neighborhood of $\theta_*^M$, and the terms in \eqref{eq:functional_if} are square integrable. 
					Define
						\begin{align}
							V_M=E_0\{(\psi_i^M)^2\}.
							\label{eq:ate_variance}
						\end{align}.
						
						Thus $V_M$ includes both estimation of the model parameter and empirical estimation of the outer expectation over $X$. 
					If $V_M>0$, then
					\begin{align}
						\sqrt{n}\left(\hat{\tau}^M - \tau_*^M\right) \xrightarrow{d} N(0, V_M).
						\label{eq:ate_clt}
					\end{align}
					If model $M$ is correctly specified, $\tau_*^M=\tau(F)$.
				\end{theorem}
				
				A feasible variance estimator is obtained from the estimated
				influence values
				\[
				\hat\psi_i^M
				=h_{\hat\theta^M}(X_i)-\hat\tau^M
				-\hat d_M^\T\hat A_M^{-1}\hat s_i^M,
				\qquad
				\hat s_i^M=\nabla_\theta\log f_M(O_i;\hat\theta^M),
				\qquad
				\hat d_M=\frac1n\sum_{i=1}^n
				\nabla_\theta h_{\hat\theta^M}(X_i),
				\]
				through
				\begin{equation}
					\hat V_M
					=\frac1n\sum_{i=1}^n
					\left(\hat\psi_i^M-\overline{\hat\psi}^{\,M}\right)^2,
					\qquad
					\overline{\hat\psi}^{\,M}=\frac1n\sum_{i=1}^n\hat\psi_i^M.
					\label{eq:ate_variance_estimator}
				\end{equation}
				
				\subsection{Mediation functionals for the Job Corps application}
				\label{subsec:mediation_functionals}
				
				{Write $X=(T,C)$, where $T$ is randomized assignment and $C$ contains
					the baseline covariates. The causal interpretation of the mediation
					summaries reported in Section~\ref{sec:app} relies on the following standard
					conditions \citep{2019Qin,Zuo2025}.
					\begin{assumption}[Mediation identification]\label{assu:mediation}
						The selected missingness model is correctly specified. Consistency and
						treatment positivity hold, together with sequential ignorability for the
						mediator--outcome relation and the absence of a post-treatment
						mediator--outcome confounder affected by treatment.
					\end{assumption}
					Let $E_C$ denote expectation over the population distribution of $C$. For
					binary $Y_1$, the natural indirect and direct effects are
					\begin{align}
						\tau_{\mathrm{NIE}}
						&=E_C\!\left[\sum_{y_1=0}^1 E(Y_2\mid T=1,Y_1=y_1,C)
						\{P(Y_1=y_1\mid T=1,C)-P(Y_1=y_1\mid T=0,C)\}\right],
						\label{eq:jobcorps_nie}\\
						\tau_{\mathrm{NDE}}
						&=E_C\!\left[\sum_{y_1=0}^1
						\{E(Y_2\mid T=1,Y_1=y_1,C)-E(Y_2\mid T=0,Y_1=y_1,C)\}
						P(Y_1=y_1\mid T=0,C)\right].
						\label{eq:jobcorps_nde}
					\end{align}
					The outer expectations are estimated by averaging the fitted conditional
					quantities over the empirical distribution of $C$. Without
					Assumption~\ref{assu:mediation}, \eqref{eq:jobcorps_nie} and
					\eqref{eq:jobcorps_nde} define model-based population functionals but do
					not identify causal natural effects.}

				\section{Graphs Outside the Main Candidate Set}
				\label{sec:outside_candidate_graphs}
				
				The main analysis deliberately restricts the candidate set to graphs in
				which at least one outcome is self-censoring. This section gives two
				examples only to delimit that scope. The first is a graph without
				self-censoring for which the substantive distribution is identifiable,
				but which is omitted because it represents a different selection problem.
				The second shows that identification is not automatic
				once both self-censoring and cross-process dependence are allowed. These
				examples are not intended as an exhaustive classification of the
				remaining DAGs.
				
				\subsection{A non-self-censoring graph outside the selection problem}
				
				Consider the graph in Figure~\ref{fig:nonself_graph}. It contains the
				cross-process arrows $R_1\to R_2$ and $Y_1\to R_2$, but neither
				$Y_1\to R_1$ nor $Y_2\to R_2$. Thus neither outcome is self-censoring.
				
				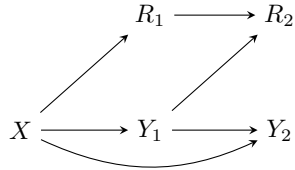
\begin{figure}[htbp]
					\centering
					\begin{tikzpicture}[scale=0.95,transform shape]
						\node (X) at (0,0) {$X$};
						\node (Yone) at (1.8,0) {$Y_1$};
						\node (Ytwo) at (3.6,0) {$Y_2$};
						\node (Rone) at (1.8,1.6) {$R_1$};
						\node (Rtwo) at (3.6,1.6) {$R_2$};
						\draw[-{stealth}] (X) -- (Yone);
						\draw[-{stealth}] (Yone) -- (Ytwo);
						\draw[-{stealth}] (X) to[out=-25,in=-155] (Ytwo);
						\draw[-{stealth}] (X) -- (Rone);
						\draw[-{stealth}] (Rone) -- (Rtwo);
						\draw[-{stealth}] (Yone) -- (Rtwo);
					\end{tikzpicture}
					\caption{A representative non-self-censoring graph outside the candidate
						set used for model selection.}
					\label{fig:nonself_graph}
				\end{figure}
				
				Its full-data law factorizes as
				\[
				f(x,y_1,y_2,r_1,r_2)
				=f(x)f(y_1\mid x)f(y_2\mid x,y_1)
				f(r_1\mid x)f(r_2\mid r_1,y_1).
				\]
				
				\begin{proposition}
					\label{prop:nonself_graph_identified}
					Suppose
					\[
					P(R_1=1\mid X=x)>0
					\]
					for every $x$ in the support of $X$, and
					\[
					P(R_2=1\mid X=x,Y_1=y_1,R_1=1)>0
					\]
					for every $(x,y_1)$ in the support of $(X,Y_1)$. Then the substantive
					full-data distribution of $(X,Y_1,Y_2)$ is identified.
				\end{proposition}
				
				\begin{proof}
					The graph implies
					\[
					R_1\indep(Y_1,Y_2)\mid X,
					\qquad
					R_2\indep Y_2\mid(X,Y_1,R_1).
					\]
					Because $X$ is fully observed,
					\[
					f(y_1\mid x)=f(y_1\mid x,R_1=1),
					\]
					and sequential positivity gives
					\[
					f(y_2\mid x,y_1)
					=f(y_2\mid x,y_1,R_1=1,R_2=1).
					\]
					Both right-hand sides are observed-data quantities, so
					$f(x)f(y_1\mid x)f(y_2\mid x,y_1)$ is identified.
				\end{proof}
				
				The proposition concerns the substantive distribution, which is the
				object needed to evaluate the conditional map $h_F(x)$ and hence the
				population functional in \eqref{eq:general_functional}. It does not require identification of
				response probabilities indexed by an unobserved value of $Y_1$ when
				$R_1=0$.
				
				This graph is omitted from the main candidate set because it has no
				self-censoring pathway, not because it is unidentified. Including it
				would broaden the scientific comparison beyond the two self-censoring
				families $M_1$ and $M_2$ studied in the paper.
				
				\subsection{A nonidentified graph in \texorpdfstring{{\purple $M_{1,2}$}}{M1,2}: panel (j)}
				
				We next consider the graph in {\purple $M_{1,2}$} shown in
				Figure~\ref{fig:graph_catalogue}(j). It contains the three missingness
				arrows
				\[
				Y_1\to R_1,\qquad R_1\to R_2,\qquad Y_1\to R_2,
				\]
				but not $Y_2\to R_2$. The following counterexample uses binary $Y_1$
				and $Y_2$. The same construction can be applied within each stratum of
				the fully observed $X$, so $X$ is suppressed. Let $P$ and
				$\widetilde P$ be the two full-data laws specified in
				Table~\ref{tab:two_full_laws}.
				
				\begin{table}[htbp]
					\centering
					\caption{Two distinct full-data laws for the graph in
						Figure~\ref{fig:graph_catalogue}(j).}
					\label{tab:two_full_laws}
					\small
					\begin{tabular}{p{0.58\textwidth}cc}
						\toprule
						Quantity & $P$ & $\widetilde P$ \\
						\midrule
						$P(Y_1=1)$ & $0.55$ & $0.60$ \\
						$P(Y_2=1\mid Y_1=0)$ & $0.30$ & $0.30$ \\
						$P(Y_2=1\mid Y_1=1)$ & $0.70$ & $0.70$ \\
						\addlinespace
						$P(R_1=1\mid Y_1=0)$ & $0.40$ & $0.45$ \\
						$P(R_1=1\mid Y_1=1)$ & $0.70$ & $77/120$ \\
						\addlinespace
						$P(R_2=1\mid R_1=1,Y_1=0)$ & $0.40$ & $0.40$ \\
						$P(R_2=1\mid R_1=1,Y_1=1)$ & $0.80$ & $0.80$ \\
						$P(R_2=1\mid R_1=0,Y_1=0)$ & $0.25$ & $27/88$ \\
						$P(R_2=1\mid R_1=0,Y_1=1)$ & $0.55$ & $363/860$ \\
						\bottomrule
					\end{tabular}
				\end{table}
				
				All probabilities lie strictly between zero and one. In each law, the
				response probability for $R_1$ varies with $Y_1$, while the response
				probability for $R_2$ varies with both $R_1$ and $Y_1$ but not with
				$Y_2$. Thus every missingness edge in
				Figure~\ref{fig:graph_catalogue}(j) is active under both laws.
				
				The observable events and their probabilities are listed in
				Table~\ref{tab:same_observed_law}. They form a partition of the observed
				sample space: both outcomes are observed when $(R_1,R_2)=(1,1)$; only
				$Y_1$ is observed when $(R_1,R_2)=(1,0)$; only $Y_2$ is observed when
				$(R_1,R_2)=(0,1)$; and neither is observed when $(R_1,R_2)=(0,0)$.
				Writing
				\[
				p_{y_1y_2}=P(Y_1=y_1,Y_2=y_2),\qquad
				a_{y_1}=P(R_1=1\mid Y_1=y_1),
				\]
				and
				\[
				b_{r_1y_1}=P(R_2=1\mid R_1=r_1,Y_1=y_1),
				\]
				the observed probabilities are
					\begin{align*}
						P(R_1 = 1,R_2 = 1,Y_1 = y_1,Y_2 = y_2)&=p_{y_1y_2}a_{y_1}b_{1y_1},\\
						P(R_1 = 1,R_2 = 0,Y_1 = y_1)&=\sum_{y_2}p_{y_1y_2}a_{y_1}(1-b_{1y_1}),\\
						P(R_1 = 0,R_2 = 1,Y_2 = y_2)&=\sum_{y_1}p_{y_1y_2}(1-a_{y_1})b_{0y_1},\\
						P(R_1 = 0,R_2 = 0)&=\sum_{y_1,y_2}p_{y_1y_2}(1-a_{y_1})(1-b_{0y_1}).
					\end{align*}

				\begin{table}[htbp]
					\centering
					\caption{The common observed-data law induced by the two full-data laws
						in Table~\ref{tab:two_full_laws}.}
					\label{tab:same_observed_law}
					\small
					\begin{tabular}{p{0.60\textwidth}cc}
						\toprule
						Observed event & $P$ & $\widetilde P$ \\
						\midrule
						$R_1=1,R_2=1,Y_1=0,Y_2=0$ & $0.050400$ & $0.050400$ \\
						$R_1=1,R_2=1,Y_1=0,Y_2=1$ & $0.021600$ & $0.021600$ \\
						$R_1=1,R_2=1,Y_1=1,Y_2=0$ & $0.092400$ & $0.092400$ \\
						$R_1=1,R_2=1,Y_1=1,Y_2=1$ & $0.215600$ & $0.215600$ \\
						$R_1=1,R_2=0,Y_1=0$ & $0.108000$ & $0.108000$ \\
						$R_1=1,R_2=0,Y_1=1$ & $0.077000$ & $0.077000$ \\
						$R_1=0,R_2=1,Y_2=0$ & $0.074475$ & $0.074475$ \\
						$R_1=0,R_2=1,Y_2=1$ & $0.083775$ & $0.083775$ \\
						$R_1=0,R_2=0$ & $0.276750$ & $0.276750$ \\
						\bottomrule
					\end{tabular}
				\end{table}
				
				The entries in each column of Table~\ref{tab:same_observed_law} sum to
				one, and the columns agree entry by entry. Hence the two distinct
				full-data laws induce the same observed-data distribution. In
				particular,
				\[
				P(Y_1=1)=0.55\ne0.60=\widetilde P(Y_1=1),
				\]
				so the substantive full-data distribution is not identified for panel
				(j) without additional restrictions. This construction keeps all three
				missingness dependencies in panel (j) active. Because panel (j) is
				nested in panel (l), the same counterexample also establishes that
				identification cannot be guaranteed for the larger graph represented by
				panel (l).

				\section{Proof of Theorem \ref{thm:ate_normality}}
				\label{pf:ate_normality}
					\begin{proof}
						A first-order expansion of \eqref{eq:functional_main} around
						$\theta_*^M$, together with the law of large numbers for the
						derivative, gives
						\begin{align}
							\sqrt n(\hat\tau^M-\tau_*^M)
							={}&\frac1{\sqrt n}\sum_{i=1}^n
							\left\{h_{\theta_*^M}(X_i)-\tau_*^M
							-d_M^\T A_M^{-1}s_i^M\right\}+o_p(1).
							\label{eq:delta_method}
						\end{align}
						The central limit theorem therefore establishes \eqref{eq:ate_clt}
						with variance $V_M=E_0\{(\psi_i^M)^2\}$.
					\end{proof}

				
				
				
				
				
				
				

				\section{Details of the Estimation via EM algorithm}\label{sec:detail_EM}
				
				The EM algorithm updates the parameter vector by maximizing the following $Q$-function:
				\begin{align}
					\label{eq:Qfunc}
					Q(\theta\mid\theta^{(t)})
					=\sum_{i=1}^nQ_i^{(t)}(\theta)
					=\sum_{i=1}^n
					E_{\theta^{(t)}}\{\ell_c^{i}(\theta)\mid O_i\},
				\end{align}
				where $\ell_c^{i}(\theta)$ denotes the contribution of the $i$th individual to the complete-data log-likelihood.  
				
				\subsection{EM Algorithm for Discrete $Y_2$}\label{subsec:DIS_EM}
				When the outcome variable $Y_2$ is discrete, we adopt a cell-count EM algorithm following \citet{Ding2014}, assuming that the covariate 
				$X$ has been discretized or its relationship with the outcome is modeled via a generalized linear model.
				We detail the procedure under model $M_1$; model $M_2$ follow analogously.

				Let $N_{x y_1 r_1 y r_2}$ denote the observed cell count. 
				A ``$+$'' subscript indicates summation over the corresponding margins; for instance, 
				$N_{x+0y1}$ sums over $Y_1$ with $R_1=0$ 
				and $R_2=1$. In the E-step, missing components are imputed via posterior distribution, yielding expected cell counts $\widetilde N^{(t)}_{x y_1 r_1 y r_2}$. In the M-step, the three model factors $f(Y_1,R_1 \mid X)$,
				$f(Y_2 \mid Y_1,X)$, and
				$P(R_2 \mid R_1,Y_2)$ are updated via maximum likelihood estimation based on
				$\widetilde N^{(t)}$.
				
				Under model $M_1$ and Assumptions \ref{assu:indep}-\ref{assu:M1_1}, the joint distribution factorizes as
				\begin{align*}
					&P_{M_1}^{(t)}(X=x,Y_1=y_1,R_1=r_1,Y_2=y_2,R_2=r_2)\\
					={}& P^{(t)}(X=x)P^{(t)}(Y_1=y_1,R_1=r_1\mid X)\\
					&\quad\times P^{(t)}(Y_2=y_2\mid X=x,Y_1=y_1)
					P^{(t)}(R_2=r_2\mid R_1=r_1,Y_2=y_2).
				\end{align*}

				Omitting $\log f_X(X)$ as it does not depend on $\theta$, 
				the complete-data log-likelihood is:
				\begin{align*}
					\ell_c(\theta^{M_1}) 
					= \sum_{x,y_1,r_1,y_2,r_2} N_{x y_1 r_1 y_2 r_2}
					\{
					\log f_1(y_1,r_1\mid x) 
					+ \log f_2(y_2\mid x,y_1) 
					+ \log P_R(r_2\mid r_1,y_2)
					\},
				\end{align*}
				where
				$$
				\begin{aligned}
					&f_1(y_1,r_1\mid x)=f(Y_1=y_1,R_1=r_1\mid X=x;\theta_{Y_1}^{M_1}),\\
					&f_2(y_2\mid x,y_1)=f(Y_2=y_2\mid X=x,Y_1=y_1;\theta_{Y_2}^{M_1}),\\
					&P_R(r_2\mid r_1,y_2)=P(R_2=r_2\mid R_1=r_1,Y_2=y_2;\theta_R^{M_1}).
				\end{aligned}
				$$
				The EM algorithm proceeds as follows.
				\begin{itemize}
					\item  \textbf{E-step:}
					Compute the expected complete-data log likelihood
					\begin{align*}
						Q(\theta^{M_1} \mid \theta^{M_1(t)})
						={}& \sum_{x,y_1,r_1,y_2,r_2}
						\widetilde{N}^{(t)}_{x y_1 r_1 y_2 r_2}
						\\[-2pt]
						&\quad\times
						\{
						\log f_1(y_1,r_1\mid x)
						+ \log f_2(y_2\mid x,y_1)
						+ \log P_R(r_2\mid r_1,y_2)
						\},
					\end{align*}
					where $\widetilde{N}^{(t)}_{xy_1r_1y_2r_2}$ is obtained across 
					four observation cases:

					\textbf{Case 1}: $R_{1i}=1,\, R_{2i}=1$ (both observed).
					
					$$
					\widetilde N^{(t)}_{x y_11y_21} = N_{x y_11 y_21}.$$

					\textbf{Case 2}: $R_{1i}=1,\, R_{2i}=0$ ($Y_1$ observed, $Y_2$ missing).
					$$
					w^{(t)}_{y_2\mid x,y_1,\,1,\,0}
					=\frac{f_2^{(t)}(y\mid x,y_1)\,P_R^{(t)}(0\mid 1,y_2)}{\sum_{u\in\{0,1\}}f_2^{(t)}(u\mid x,y_1)\,P_R^{(t)}(0\mid 1,u)}.
					$$
					$$
					\widetilde N^{(t)}_{x y_11y_20}
					= N_{x y_11+0}\, w^{(t)}_{y_2\mid x,y_1,\,1,\,0},
					\qquad y_2\in\{0,1\}.
					$$
					\textbf{Case 3}: $R_{1i}=0,\, R_{2i}=1$ ($Y_1$ missing, $Y_2$ observed).
					$$
					w^{(t)}_{y_1 \mid x,0,y_2,1}
					= \frac{
						f_1^{(t)}(y_1,0\mid x)\,f_2^{(t)}(y_2\mid x,y_1)
					}{
						\sum_{y_1' \in \{0,1\}} f_1^{(t)}(y_1',0\mid x)\,f_2^{(t)}(y_2\mid x,y_1')
					}.
					$$
					$$
					\widetilde N^{(t)}_{x y_1 0 y_2 1}
					= N_{x + 0 y_2 1}\, w^{(t)}_{y_1\mid x,\,0,\,y_2,\,1},
					\qquad a\in\{0,1\}.
					$$

					\textbf{Case 4}: $R_{1i}=0,\, R_{2i}=0$ (both missing).
					
					$$
					w^{(t)}_{y_1,y_2\mid x,\,0,\,0}
					=\frac{f_1^{(t)}(y_1,0\mid x)\,f_2^{(t)}(y_2\mid x,y_1)\,P_R^{(t)}(0\mid 0,y_2)}
					{\sum_{y_1'\in\{0,1\}}\sum_{u\in\{0,1\}}f_1^{(t)}(y_1',0\mid x)\,f_2^{(t)}(u\mid x,y_1')\,P_R^{(t)}(0\mid 0,u)},
					$$
					$$
					\widetilde N^{(t)}_{x y_10 y_20} = N_{x+0+0}\, w^{(t)}_{y_1,y_2\mid x,\,0,\,0}\quad(y_1,y_2\in\{0,1\}).
					$$
					
					\item \textbf{M-Step:}
					Aggregate imputed counts across all cases:
					$$
					\widetilde N^{(t)}_{x y_1 r_1 y_2 r_2}
					=\widetilde N^{(t)}_{x y_1 1 y_2 1}
					+\widetilde N^{(t)}_{x y_1 1 y_2 0}
					+\widetilde N^{(t)}_{x y_1 0 y_2 1}
					+\widetilde N^{(t)}_{x y_1 0 y_2 0},
					$$
					and update parameters via maximum likelihood:
					\begin{align*}
						\hat f_1^{(t+1)}(y_1,r_1\mid x)=& \widehat P^{(t+1)}(Y_1=y_1,R_1=r_1\mid X=x)
						=\frac{\sum_{y_2 ,r_2}\widetilde N^{(t)}_{x y_1 r_1 y_2  r_2}}
						{\sum_{y_1',r_1',y_2 ,r_2}\widetilde N^{(t)}_{x y_1' r_1' y_2  r_2}}
						=\frac{\widetilde N^{(t)}_{x y_1 r_1 + +}}{\widetilde N^{(t)}_{x + + + +}},\\
						\hat f_2^{(t+1)}(y_2 \mid x,y_1)= & \widehat P^{(t+1)}(Y_2=y_2\mid X=x,Y_1=y_1)
						=\frac{\sum_{r_1,r_2}\widetilde N^{(t)}_{x y_1 r_1 y_2 r_2}}
						{\sum_{r_1,r_2,y_2'}\widetilde N^{(t)}_{x y_1 r_1 y_2' r_2}}
						=\frac{\widetilde N^{(t)}_{x y_1 + y_2 +}}{\widetilde N^{(t)}_{x y_1 + + +}},\\
						\hat P_R^{(t+1)}(r_2\mid r_1,y_2)= & \widehat P^{(t+1)}(R_2=r_2\mid R_1=r_1,Y_2=y_2) 
						=\frac{\sum_{x,y_1}\widetilde N^{(t)}_{x y_1 r_1 y_2 r_2}}
						{\sum_{x,y_1,r_2'}\widetilde N^{(t)}_{x y_1 r_1 y_2 r_2'}}
						=\frac{\widetilde N^{(t)}_{+ + r_1 y_2 r_2}}{\widetilde N^{(t)}_{+ + r_1 y_2 +}}.
					\end{align*}
					
				\end{itemize}
				These steps iterate until convergence.
				
				For model $M_2$, the algorithm is identical in structure, 
				replacing $P^{(t)}(R_2\!=\!r_2\mid R_1\!=\!r_1,Y_2\!=\!y_2)$ 
				with $P^{(t)}(R_2\!=\!r_2\mid Y_1\!=\!y_1,Y_2\!=\!y_2)$. 
				All other steps follow analogously and are omitted.

				\subsection{EM Algorithm for Continuous $Y_2$}\label{subsec:CON_EM}
				When $Y_2$ is continuous, the E-step requires integration over the missing components rather than summation over finite cells. We approximate these conditional expectations by importance sampling, which yields a Monte Carlo EM algorithm. With a finite Monte Carlo sample, monotonicity of the observed-data likelihood is not guaranteed. The asymptotic theory in this paper is therefore stated for the observed-data MLE rather than for a finite-simulation MCEM iterate.

				For each model, we specify parametric models for three components: the joint treatment-missingness model $f(Y_{1i}, R_{1i} \mid X_i; \theta_{Y_1})$, the outcome model $f(Y_{2i} \mid X_i, Y_{1i}; \theta_{Y_2})$, and the outcome missingness model $f(R_{2i} \mid R_{1i}, Y_{2i}; \theta_R)$ under model $M_1$ or $f(R_{2i} \mid Y_{1i}, Y_{2i}; \theta_R)$ under model $M_2$. The full parameter vectors are $\theta^{M_1}=(\theta_{Y_1}^{M_1},\theta_{Y_2}^{M_1},\theta_R^{M_1})$ for model $M_1$ and $\theta^{M_2}=(\theta_{Y_1}^{M_2},\theta_{Y_2}^{M_2},\theta_R^{M_2})$ for model $M_2$.
				
				Under Assumptions~\ref{assu:indep} -- \ref{assu:M1_1}, the complete-data log-likelihood function under model $M_1$ is
				\begin{equation}
					\label{eq:ind-ll-F}
					\ell_c(\theta^{M_1})
					= \sum_{i=1}^n \{
					\log f(Y_{1i},R_{1i} \mid X_i;\theta_{Y_1}^{M_1})
					+ \log f(Y_{2i} \mid X_i,Y_{1i};\theta_{Y_2}^{M_1})
					+ \log f(R_{2i} \mid R_{1i},Y_{2i};\theta_R^{M_1})
					\}.
				\end{equation}
				where the term $\log f(X_i)$ is omitted since it does not depend on $\theta$. The observed-data  likelihood is
				\begin{align*}
					L_{\mathrm{obs}}(\theta^{M_1})=\prod_{i=1}^n \int\int  f_{M_1}(Y_{1i},Y_{2i},R_{1i},R_{2i} \mid X_i;\theta^{M_1}) dY_{1i}^{mis}dY_{2i}^{mis},
				\end{align*}
				where integration is over missing components only, with observed values substituted directly.
				We detail the EM procedure under 
				model $M_1$; model $M_2$ follows with $f(R_{2i}\mid R_{1i},Y_{2i};\theta_R^{M_1})$ 
				replaced by $f(R_{2i}\mid Y_{1i},Y_{2i};\theta_R^{M_2})$.
				
				The EM algorithm proceeds as follows.

				\begin{itemize}
					
					\item 
					\textbf{E-step:} 
					Compute the expected complete-data log-likelihood
					\begin{align*}
						Q\bigl(\theta^{M_1} \mid \theta^{M_1(t)}\bigr)
						= E_{\theta^{M_1(t)}}\!\left\{\ell_c(\theta^{M_1})
						\,\middle|\, O_1,\ldots,O_n\right\}
						= \sum_{i=1}^n \widetilde Q_i^{(t)}(\theta^{M_1}),
					\end{align*} 
					where $\widetilde Q_i^{(t)}(\theta^{M_1})$ is unit $i$'s contribution to the $Q$-function, determined by its observation pattern.
					
					\textbf{Case 1}: $R_{1i} = R_{2i} = 1$(both observed).
					No imputation is required.
					\begin{align*}
						\widetilde{Q_i}^{(t)}(\theta^{M_1})
						&=\log f(Y_{1i},R_{1i}\mid X_i;\theta_{Y_1}^{M_1})
						+\log f(Y_{2i}\mid X_i,Y_{1i};\theta_{Y_2}^{M_1})
						+\log P(R_{2i}\mid R_{1i},Y_{2i};\theta_R^{M_1}).
					\end{align*}
					
					
					\textbf{Case 2}: $R_{1i} = 1, R_{2i} = 0$ ($Y_1$ observed, $Y_2$ missing). The contribution to the $Q$-function is
					\begin{align*}
						\widetilde{Q}_i^{(t)}(\theta^{M_1})
						&=\log f(Y_{1i},R_{1i}\mid X_i;\theta_{Y_1}^{M_1})\\
						&\quad+ \mathbb{E}_{Y_{2i} \mid X_i, Y_{1i}, R_{1i}=1, R_{2i}=0; \theta^{M_1(t)}}\{\log f(Y_{2i} \mid X_i,Y_{1i};\theta_{Y_2}^{M_1})+ \log P(R_{2i} \mid R_{1i},Y_{2i};\theta_R^{M_1})\}
					\end{align*}
					where the posterior distribution of $Y_{2i}$ is
					\begin{align*}
						f(Y_{2i}\mid X_i,Y_{1i},R_{1i}=1,R_{2i}=0;\theta^{M_1(t)})
						&\propto f(Y_{2i}\mid X_i,Y_{1i};\theta_{Y_2}^{M_1(t)}) \cdot P(R_{2i}=0\mid R_{1i}=1,Y_{2i};\theta_R^{M_1(t)}).
					\end{align*}
					Since direct integration is computationally challenging, we use importance sampling.
					Using $B_{\mathrm{MC}}$ Monte Carlo draws
					$\{Y_{2i}^{(m)}\}_{m=1}^{B_{\mathrm{MC}}}$ from a proposal distribution
					$g(Y_{2i})$, compute the importance weights
					\begin{align*}
						w_i^{(m)} \propto \frac{f(Y_{2i}^{(m)}\mid X_i,Y_{1i};\theta_{Y_2}^{M_1(t)}) \cdot P(R_{2i}=0\mid R_{1i}=1,Y_{2i}^{(m)};\theta_R^{M_1(t)})}{g(Y_{2i}^{(m)})},
					\end{align*}
					the conditional expectation is approximated by the weighted samples.

					Case 3: $R_{1i} = 0, R_{2i} = 1$($Y_1$ missing, $Y_2$ observed).
					
					Since $Y_{1i}\in\{0,1\}$ is binary, the contribution to the $Q$-function is a finite sum:
					\begin{align*}
						\widetilde{Q}_i^{(t)}(\theta^{M_1})
						={} & \sum_{y_1=0}^1 P(Y_{1i}=y_1 \mid X_i, R_{1i}=0, Y_{2i}, R_{2i}=1; \theta^{M_1(t)}) \\
						& \quad \times \{\log f(Y_{1i}=y_1, R_{1i}=0 \mid X_i;\theta_{Y_1}^{M_1}) + \log f(Y_{2i} \mid X_i,Y_{1i}=y_1;\theta_{Y_2}^{M_1}) \\
						& \qquad + \log P(R_{2i}=1 \mid R_{1i}=0,Y_{2i};\theta_R^{M_1}) \},
					\end{align*}
					where the posterior distribution of $Y_1$ is
					\begin{align*}
						&P(Y_{1i}=y_1 \mid X_i,R_{1i}=0,Y_{2i},R_{2i}=1;\theta^{M_1(t)})\\
						&=\frac{
							P(Y_{1i}=y_1,R_{1i}\mid X_i;\theta_{Y_1}^{M_1(t)})
							f(Y_{2i}\mid X_i,Y_{1i}=y_1;\theta_{Y_2}^{M_1(t)})
							P(R_{2i}\mid R_{1i},Y_{2i};\theta_R^{M_1(t)})
						}{
							\sum_{y_1'\in\{0,1\}}
							P(Y_{1i}=y_1',R_{1i}\mid X_i;\theta_{Y_1}^{M_1(t)})
							f(Y_{2i}\mid X_i,Y_{1i}=y_1';\theta_{Y_2}^{M_1(t)})
							P(R_{2i}\mid R_{1i},Y_{2i};\theta_R^{M_1(t)})
						}.
					\end{align*}
					

					
					Case 4: $R_{1i} = R_{2i} = 0$ (both missing).
					The contribution to the $Q$-function is
					\begin{align*}
						\widetilde{Q}_i^{(t)}(\theta^{M_1})
						&= \sum_{y_1=0}^1 \int f(Y_{1i}=y_1, Y_{2i}=y_2 \mid X_i, R_{1i}=0, R_{2i}=0;\theta^{M_1(t)}) \\
						&\quad \times \{\log f(Y_{1i}=y_1, R_{1i}=0 \mid X_i;\theta_{Y_1}^{M_1}) + \log f(Y_{2i}=y_2 \mid X_i,Y_{1i}=y_1;\theta_{Y_2}^{M_1}) \\
						&\qquad + \log P(R_{2i}=0 \mid R_{1i}=0,Y_{2i}=y_2;\theta_R^{M_1})\} \, dy_2
					\end{align*}
					
					The joint posterior distribution is
					\begin{align*}
						&f(Y_{1i}=y_1,Y_{2i}=y_2\mid X_i, R_{1i}=0, R_{2i}=0;\theta^{M_1(t)})\\
						\propto & f(Y_{1i}=y_1,R_{1i}=0\mid X_i;\theta_{Y_1}^{M_1(t)}) 
						f(Y_{2i}=y_2\mid X_i,Y_{1i}=y_1;\theta_{Y_2}^{M_1(t)})\,P(R_{2i}=0\mid R_{1i}=0,Y_{2i}=y_2;\theta_R^{M_1(t)})
					\end{align*}
					
					As in case 2, we use importance sampling: draw
					$\{(Y_{1i}^{(m)},Y_{2i}^{(m)})\}_{m=1}^{B_{\mathrm{MC}}}$
					from a proposal distribution, compute joint importance weights $w_i^{(m)}(y_1,y_2)$, and use the weighted samples to compute the conditional expectation.
					

					\item \textbf{M-step}
					
					Maximize  $Q(\theta^{M_1} \mid \theta^{M_1(t)})$ to obtain:
					\begin{align*}
						\theta^{M_1(t+1)} = \arg\max_{\theta^{M_1}} Q(\theta^{M_1} \mid \theta^{M_1(t)})
						= \arg\max_{\theta^{M_1}} \sum_{i=1}^n \widetilde{Q}_i^{(t)}(\theta^{M_1})
					\end{align*}
					
				\end{itemize}

				
				The familiar Jensen-inequality argument for a non-decreasing
				observed-data likelihood applies to exact EM, but not directly to the
				finite-simulation approximation above. A separate asymptotic analysis
				of MCEM would require additional conditions on the Monte Carlo error;
				such conditions are not imposed here.

				\section{Regularity Conditions for Model Selection Procedure}\label{sm_sec:regularity}
				We impose the following regularity conditions to derive the asymptotic distributions of the test statistics in the model selection procedure.

				\begin{assumption}\label{assu:regularity} 
					Let $F_0$ denote the true distribution of the observed data $O_i$, with
					density $f_0$ with respect to a $\sigma$-finite dominating measure
					$\nu$. For models $M_1$ and $M_2$, let
					$f_{M_1}(o;\theta^{M_1})$ and $f_{M_2}(o;\theta^{M_2})$ denote the
					observed-data densities, with parameter spaces
					$\Theta_{M_1}\subset\mathbb{R}^p$ and
					$\Theta_{M_2}\subset\mathbb{R}^q$, respectively.
					Assume:
					
					\begin{enumerate}[label=(\roman*)]
						
						\item (Support Consistency) The observed data vectors $O_i$, $i=1,\ldots,n$, are independent and identically distributed under $F_0$.
						The densities $f_0(o)$, $f_{M_1}(o;\theta^{M_1})$, and $f_{M_2}(o;\theta^{M_2})$ share a common support $K^0$, on which the model densities are strictly positive for all admissible parameter values.
						
						\item (Compactness and Continuity)
						The parameter spaces $\Theta_{M_1}$ and $\Theta_{M_2}$ are compact. For $F_0$-almost every $o\in K^0$,
						$f_{M_1}(o;\theta^{M_1})
						$ and $f_{M_2}(o;\theta^{M_2})$
						are continuous in $\theta^{M_1}$ and $\theta^{M_2}$, respectively.
						
						\item (Integrability and Pseudo-Truth Existence)
						For any $\theta^{M_1} \in \Theta_{M_1}$ and $\theta^{M_2} \in \Theta_{M_2}$,
						$ 
						|\log f_{M_1}(o;\theta^{M_1})|$
						and $
						|\log f_{M_2}(o;\theta^{M_2})|
						$ 
						are uniformly dominated by an integrable function. The expected
						log-likelihood under each model attains a unique maximum
						at the pseudo-truth values $\theta_*^{M_1}$ and $\theta_*^{M_2}$, and the MLEs $\hat\theta^{M_1}$ and $\hat\theta^{M_2}$ exist with probability approaching one.
						
						\item (Smoothness)
						There exist open neighbourhoods
						$\mathcal{N}_{M_1} \ni \theta_*^{M_1}$ and $\mathcal{N}_{M_2} \ni \theta_*^{M_2}$, such that for almost all $o$,
						$\log f_{M_1}(o;\theta^{M_1})$ and $\log f_{M_2}(o;\theta^{M_2})$
						are twice continuously differentiable in $\theta^{M_1} \in \mathcal{N}_{M_1}$ and $\theta^{M_2} \in \mathcal{N}_{M_2}$,
						with first and second derivatives dominated by integrable functions, respectively.
						The following matrices exist and are well defined:
						\begin{align*}
							B_{M_1} &= E_0\!\left\{
							\frac{\partial\log f_{M_1}(o;\theta_*^{M_1})}{\partial\theta^{M_1}}
							\cdot
							\frac{\partial\log f_{M_1}(o;\theta_*^{M_1})}{\partial(\theta^{M_1})^\T}
							\right\},\\
							B_{M_2} &= E_0\!\left\{
							\frac{\partial\log f_{M_2}(o;\theta_*^{M_2})}{\partial\theta^{M_2}}
							\cdot
							\frac{\partial\log f_{M_2}(o;\theta_*^{M_2})}{\partial(\theta^{M_2})^\T}
							\right\},\\
							B_{M_1M_2} &= E_0\!\left\{
							\frac{\partial\log f_{M_1}(o;\theta_*^{M_1})}{\partial\theta^{M_1}}
							\cdot
							\frac{\partial\log f_{M_2}(o;\theta_*^{M_2})}{\partial(\theta^{M_2})^\T}
							\right\}.
						\end{align*}
						Set $B_{M_2M_1}=B_{M_1M_2}^{\T}$.
						\item (Regularity at Pseudo-Truth Points)
						$\theta_*^{M_1}$ and $\theta_*^{M_2}$ are interior points of  $\Theta_{M_1}$ and $\Theta_{M_2}$.
						The Hessian matrices
						\begin{align*}
							A_{M_1} = E_0\!\left\{
							\frac{\partial^2\log f_{M_1}(o;\theta_*^{M_1})}{\partial\theta^{M_1}\partial(\theta^{M_1})^\T}
							\right\},
							A_{M_2} = E_0\!\left\{
							\frac{\partial^2\log f_{M_2}(o;\theta_*^{M_2})}{\partial\theta^{M_2}\partial(\theta^{M_2})^\T}
							\right\},
						\end{align*}
						are both nonsingular.
						
						\item (Finite second moment) The squared log-densities $\left|\log f_{M_1}(o;\theta^{M_1})\right|^2$ and $\left|\log f_{M_2}(o;\theta^{M_2})\right|^2$ are dominated by $F_0$-integrable functions that do not depend on $\theta^{M_1}$ and $\theta^{M_2}$, respectively.

						\item (Dominated Integrability for variance estimation)
						Define $m(o;\theta^{M_1},\theta^{M_2})=\log\cfrac{f_{M_1}(o;\theta^{M_1})}
						{f_{M_2}(o;\theta^{M_2})}.$
						For $F_0$-almost every $o\in K^0$ and all 
						$(\theta^{M_1},\theta^{M_2})\in\Theta_{M_1}\times\Theta_{M_2}$, the 
						functions
						
						\begin{align*}
							&\left|m(o;\theta^{M_1},\theta^{M_2})\right|\cdot
							\left\|\frac{\partial^2\log f_{M_1}(o;\theta^{M_1})}{\partial\theta^{M_1}\partial(\theta^{M_1})^\T}\right\|,~\left|m(o;\theta^{M_1},\theta^{M_2})\right|\cdot
							\left\|\frac{\partial^2\log f_{M_2}(o;\theta^{M_2})}{\partial\theta^{M_2}\partial(\theta^{M_2})^\T}\right\|
						\end{align*}
						are dominated by $F_0$-integrable functions independent of
						$\theta^{M_1}$ and $\theta^{M_2}$.
					\end{enumerate}
				\end{assumption}

				Assumption \ref{assu:regularity} adapts the regularity conditions of \citet{Vuong1989} to the MNAR setting. Under this framework, the model selection problem is naturally formulated as a hypothesis-testing problem based on the log-likelihood ratio of the two competing models.  
				Conditions (i)--(ii) ensure a well-defined observed-data distribution, common support across models, and existence of maximizers via compactness and continuity. Condition (iii) provides the uniform integrability needed for the law of large numbers and establishes the existence and uniqueness of the pseudo-true parameters $\theta_*^{M_1}$ and $\theta_*^{M_2}$; under misspecification, each pseudo-true parameter indexes the member of the corresponding model that minimizes the Kullback--Leibler divergence from $F_0$. Conditions (iv)--(v) supply the smoothness and nonsingularity needed for standard M-estimation expansions. Conditions (vi)--(vii) provide the moment and domination conditions used by the two model-selection tests. In particular, together with the overlap null $\omega_*^2=0$, they yield $\hat\omega_n^2=O_p(n^{-1})$ and a nondegenerate limit for $n\hat\omega_n^2$.

				\section{Proof of Model Selection Test}\label{pf:MS}
				
				\subsection{Proof of Variance Consistency}\label{pf:lemma1}
				\begin{proof}[The proof of Lemma \ref{lem:variance_consistency}]
					The proof proceeds in two steps.
					
					\noindent\textbf{Step (i):} $\hat\omega_n^2 \xrightarrow{a.s.} \omega_*^2$.
					Given assumptions \ref{assu:regularity}(i)-(iii) and (vi), from Jennrich's uniform strong law of large numbers (\citet{Jennrich1969}, Theorem 2), we have:
					$$
					\frac{1}{n}\sum_{i=1}^{n}\{\ell_{i}^{M_1}(\theta^{M_1}) - \ell_{i}^{M_2}(\theta^{M_2})\} \xrightarrow{a.s.} E_0\{\ell_{i}^{M_1}(\theta_{*}^{M_1}) - \ell_{i}^{M_2}(\theta_{*}^{M_2})\}
					$$
					$$
					\frac{1}{n}\sum_{i=1}^{n}\{\ell_{i}^{M_1}(\theta^{M_1}) - \ell_{i}^{M_2}(\theta^{M_2})\}^{2} \xrightarrow{a.s.} E_0\left[\left\{\ell_{i}^{M_1}(\theta_{*}^{M_1}) - \ell_{i}^{M_2}(\theta_{*}^{M_2})\right\}^{2}\right]
					$$
					uniformly over $(\theta^{M_1},\theta^{M_2})\in
					\Theta_{M_1}\times\Theta_{M_2}$.

					By Assumption \ref{assu:regularity}(iii), the ML estimators $\hat{\theta}^{M_1}$ and $\hat{\theta}^{M_2}$ converge strongly to the pseudo-truths $\theta_*^{M_1}$ and $\theta_*^{M_2}$, respectively.
					
					Evaluating these uniform limits at the strongly consistent maximum
					likelihood estimators gives
					$$
					\frac{1}{n}\sum_{i=1}^{n}\{\ell_{i}^{M_1}(\hat\theta^{M_1}) - \ell_{i}^{M_2}(\hat\theta^{M_2})\} \xrightarrow{a.s.} E_0\{\ell_{i}^{M_1}(\theta_*^{M_1}) - \ell_{i}^{M_2}(\theta_*^{M_2})\}
					$$
					$$
					\frac{1}{n}\sum_{i=1}^{n}\{\ell_{i}^{M_1}(\hat\theta^{M_1}) - \ell_{i}^{M_2}(\hat\theta^{M_2})\}^2 \xrightarrow{a.s.} E_0\left[\{\ell_{i}^{M_1}(\theta_*^{M_1}) - \ell_{i}^{M_2}(\theta_*^{M_2})\}^2\right]
					$$
					
					Let
					\[
					m_i=\ell_i^{M_1}(\hat\theta^{M_1})-
					\ell_i^{M_2}(\hat\theta^{M_2}),
					\qquad
					m_i^*=\ell_i^{M_1}(\theta_*^{M_1})-
					\ell_i^{M_2}(\theta_*^{M_2}).
					\]
					Then
					\begin{align*}
						\hat\omega_n^2
						&=\frac1n\sum_{i=1}^n m_i^2
						-\left(\frac1n\sum_{i=1}^n m_i\right)^2\\
						&\xrightarrow{a.s.}E_0\{(m_i^*)^2\}-\{E_0(m_i^*)\}^2\\
						&=\operatorname{Var}\!\left\{\log
						\frac{f_{M_1}(O_i;\theta_*^{M_1})}
						{f_{M_2}(O_i;\theta_*^{M_2})}\right\}
						=\omega_*^2.
					\end{align*}
					
					\textbf{Step (ii): Limit of $\tilde{\omega}_{n}^{2}$}
					
					From Step (i),$$
					\tilde{\omega}_{n}^{2} = \frac{1}{n}\sum_{i=1}^{n}\{\ell_{i}^{M_1}(\hat\theta^{M_1}) - \ell_{i}^{M_2}(\hat\theta^{M_2})\}^{2} \xrightarrow{a.s.} E_0\left[\left\{\ell_{i}^{M_1}(\theta_*^{M_1}) - \ell_{i}^{M_2}(\theta_*^{M_2})\right\}^{2}\right].
					$$
					
					By variance decomposition,
					$$
					E_0\left[\left\{\ell_{i}^{M_1}(\theta_*^{M_1}) - \ell_{i}^{M_2}(\theta_*^{M_2})\right\}^{2}\right] = \mathrm{Var}\!\left\{\log\frac{f_{M_1}(O_i;\theta_*^{M_1})}{f_{M_2}(O_i;\theta_*^{M_2})}\right\} + \left[E_0\{\ell_{i}^{M_1}(\theta_*^{M_1}) - \ell_{i}^{M_2}(\theta_*^{M_2})\}\right]^{2}
					$$
					
					Therefore,
					$$
					\tilde{\omega}_{n}^{2} \xrightarrow{a.s.} \omega_*^2 + \left[E_0\{\ell_{i}^{M_1}(\theta_*^{M_1}) - \ell_{i}^{M_2}(\theta_*^{M_2})\}\right]^{2}.
					$$
					
					This completes the proof.
				\end{proof}
				
				\subsection{Proof of Limit Distribution of Variance Test Statistics}\label{pf:thm3}
				
				\begin{proof}[Proof of Lemma \ref{lemma:equivalence_condition}]
					Write
					\[
					m_*(O_i)=\log\frac{f_{M_1}(O_i;\theta_*^{M_1})}
					{f_{M_2}(O_i;\theta_*^{M_2})}.
					\]
					If $\omega_*^2=\operatorname{Var}_0\{m_*(O_i)\}=0$, then
					$m_*(O_i)=c$ $F_0$-almost surely for some constant $c$. The common
					support and positivity condition in Assumption~\ref{assu:regularity}(i)
					then gives
					\[
					f_{M_1}(o;\theta_*^{M_1})=e^c f_{M_2}(o;\theta_*^{M_2})
					\]
					almost everywhere on $K^0$. Integrating both densities over their common
					support yields $e^c=1$, so the two densities agree almost everywhere.
					Conversely, equality of the two densities implies $m_*(O_i)=0$ almost
					surely and hence $\omega_*^2=0$.
				\end{proof}
				
				Before proving Theorem~\ref{thm:variance_distribution}, we state a
				standard result on quadratic forms of normal random vectors
				\citep{imhof1961computing}.
				\begin{lemma}\label{lemma:s1}
					Let $Z\sim N(0,\Omega)$ be an $m$-dimensional Gaussian vector, with
					$\operatorname{rank}(\Omega)\le m$, and let $U$ be a real symmetric
					$m\times m$ matrix. Then
					\begin{equation}\label{eq:wchiq}
						Z^{\T}UZ\ \overset{d}{=}\ \sum_{j=1}^{m}\lambda_j\chi_{1,j}^2,
					\end{equation}
					where $\{\chi_{1,j}^2\}_{j=1}^m$ are independent chi-squared random
					variables with one degree of freedom and $\{\lambda_j\}_{j=1}^m$ are
					the eigenvalues of $U\Omega$, including any zero eigenvalues.
				\end{lemma}

				The proof derives the limiting distribution of $\tilde\omega_n^2$ 
				with the following Lemma \ref{lem:variance_equivalence} first, then transfers the result to $\hat\omega_n^2$ via 
				$n(\hat\omega_n^2 - \tilde\omega_n^2) = o_p(1)$ and Slutsky's theorem.
				\begin{lemma}[Asymptotic Equivalence of Variance Estimates]
					\label{lem:variance_equivalence}
					Under Assumption \ref{assu:regularity} and $\mathcal{H}_0^{1}: \omega_*^2 = 0$:
					$$
					n\tilde{\omega}_{n}^{2} = n\hat{\omega}_{n}^{2} + o_p(1)
					$$
				\end{lemma}

				\begin{proof}
					By definitions of $\hat{\omega}_{n}^{2}$, $\tilde{\omega}_{n}^{2}$, and $LR_n(\hat{\theta}^{M_1}, \hat{\theta}^{M_2})$:
					\begin{equation}
						n\tilde{\omega}_{n}^{2} = n\hat{\omega}_{n}^{2} + \frac{1}{n}\{LR_n(\hat{\theta}^{M_1}, \hat{\theta}^{M_2})\}^{2}
						\label{eq:asy}
					\end{equation}
					It suffices to show that the last term in \eqref{eq:asy} is $o_p(1)$.
					Under $\mathcal H_0^1$, Lemma~\ref{lemma:equivalence_condition} gives
					\[
					LR_n(\theta_*^{M_1},\theta_*^{M_2})=0
					\quad\text{almost surely}.
					\]
					For each model, a second-order likelihood expansion around its
					pseudo-true parameter gives
					\[
					\sum_{i=1}^n
					\{\ell_i^{M_j}(\hat\theta^{M_j})
					-\ell_i^{M_j}(\theta_*^{M_j})\}
					=O_p(1),
					\qquad j=1,2,
					\]
					because the sample score is $O_p(\sqrt n)$, the estimation error is
					$O_p(n^{-1/2})$, and the sample Hessian is $O_p(n)$. Consequently,
					\[
					LR_n(\hat\theta^{M_1},\hat\theta^{M_2})=O_p(1),
					\]
					and therefore
					\[
					\frac{1}{n}\{LR_n(\hat\theta^{M_1},\hat\theta^{M_2})\}^2
					=O_p(n^{-1})=o_p(1).
					\]
					Substitution into \eqref{eq:asy} proves
					\[
					n\tilde\omega_n^2=n\hat\omega_n^2+o_p(1).
					\]
				\end{proof}
				
				\begin{proof}[Proof of Theorem \ref{thm:variance_distribution}]
					Stack the two parameter vectors as
					\[
					\vartheta=\{(\theta^{M_1})^{\T},(\theta^{M_2})^{\T}\}^{\T},
					\qquad
					\hat\vartheta=\{(\hat\theta^{M_1})^{\T},(\hat\theta^{M_2})^{\T}\}^{\T},
					\qquad
					\vartheta_*=\{(\theta_*^{M_1})^{\T},(\theta_*^{M_2})^{\T}\}^{\T},
					\]
					and write
					\[
					d_i(\vartheta)=\ell_i^{M_1}(\theta^{M_1})-
					\ell_i^{M_2}(\theta^{M_2}),
					\qquad
					g_i=\nabla_\vartheta d_i(\vartheta_*)
					=
					\begin{pmatrix}
						\nabla_{\theta^{M_1}}\ell_i^{M_1}(\theta_*^{M_1})\\
						-\nabla_{\theta^{M_2}}\ell_i^{M_2}(\theta_*^{M_2})
					\end{pmatrix}.
					\]
					Under $\mathcal H_0^1$, Lemma~\ref{lemma:equivalence_condition} implies
					$d_i(\vartheta_*)=0$ almost surely. A first-order expansion of
					$d_i(\hat\vartheta)$, together with the derivative-envelope conditions
					in Assumption~\ref{assu:regularity}, therefore yields
					\begin{equation}\label{eq:variance_quadratic_expansion}
						n\tilde\omega_n^2
						=
						\{\sqrt n(\hat\vartheta-\vartheta_*)\}^{\T}
						\left(\frac1n\sum_{i=1}^n g_i g_i^{\T}\right)
						\{\sqrt n(\hat\vartheta-\vartheta_*)\}
						+o_p(1).
					\end{equation}
					The law of large numbers gives
					\[
					\frac1n\sum_{i=1}^n g_i g_i^{\T}
					\xrightarrow{p}
					V
					=
					\begin{pmatrix}
						B_{M_1} & -B_{M_1M_2}\\
						-B_{M_2M_1} & B_{M_2}
					\end{pmatrix}.
					\]
					Combining the joint M-estimation expansion and block-matrix calculation
					in Lemma~A of \citet{Vuong1989} with
					\eqref{eq:variance_quadratic_expansion} shows directly that the
					limiting quadratic form is a weighted sum of independent $\chi_1^2$
					variables. Its weights are $\{\lambda_j^2\}_{j=1}^{p+q}$, where
					\(\{\lambda_j\}\) are the eigenvalues of the matrix \(W\) stated in the
					theorem. Hence
					\[
					P(n\tilde\omega_n^2\le t)
					\longrightarrow M_{p+q}(t;\lambda^2).
					\]
					Finally, Lemma~\ref{lem:variance_equivalence} gives
					$n\hat\omega_n^2=n\tilde\omega_n^2+o_p(1)$, so Slutsky's theorem yields
					the asserted limit for $\hat T_n=n\hat\omega_n^2$.
				\end{proof}


				\subsection{Proof of Theorem \ref{thm:lr_normal}: Normal Convergence of the LR Estimator}\label{pf:thm4}
				
				\begin{proof}
					A second-order expansion of each maximized log-likelihood around its
					pseudo-true parameter, together with the score equations and
					$\sqrt n$-consistency of the MLEs, gives
					\begin{align}
						\frac{1}{\sqrt{n}}LR_n(\hat\theta^{M_1},\hat\theta^{M_2})
						= \frac{1}{\sqrt{n}}LR_n(\theta_*^{M_1},\theta_*^{M_2}) + o_p(1).
						\label{eq:lr_replace}
					\end{align}
					Under $\mathcal H_0^2$, the summands have mean zero, and
					$\omega_*^2>0$ by assumption. Assumption~\ref{assu:regularity}(vii) and
					the central limit theorem therefore give
					\begin{align}
						\frac{1}{\sqrt{n}}LR_n(\theta_*^{M_1},\theta_*^{M_2})
						= \frac{1}{\sqrt{n}}\sum_{i=1}^n m_i(\theta_*^{M_1},\theta_*^{M_2})
						\xrightarrow{d} N(0,\omega_*^2).
						\label{eq:clt_lr}
					\end{align}
					
					Combining \eqref{eq:lr_replace}, \eqref{eq:clt_lr}, 
					and Slutsky's theorem,
					$$\frac{1}{\sqrt{n}}LR_n(\hat\theta^{M_1},\hat\theta^{M_2})
					\xrightarrow{d} N(0,\omega_*^2).$$
					By Lemma~\ref{lem:variance_consistency}(i), 
					$\hat\omega_n^2\xrightarrow{a.s.}\omega_*^2>0$, so 
					$\hat\omega_n\xrightarrow{a.s.}\omega_*>0$. 
					Slutsky's theorem gives
					\begin{align}
						\hat V_n
						= \frac{LR_n(\hat\theta^{M_1},\hat\theta^{M_2})}{\sqrt{n}\,\hat\omega_n}
						= \frac{\frac{1}{\sqrt{n}}LR_n(\hat\theta^{M_1},\hat\theta^{M_2})}
						{\hat\omega_n}
						\xrightarrow{d} \frac{N(0,\omega_*^2)}{\omega_*}
						= N(0,1).
					\end{align}
					This completes the proof.
				\end{proof}

				\section{Proofs for Section~\ref{sec:post_selection}}
				\label{sec:proof_post_selection}
				
				The joint expansion below is needed only for the selection-adjusted analysis in Appendix~\ref{app:selection_normal_boundary}. For a sequence of data-generating laws, attach the subscript $n$ to the corresponding population quantities and define
				\begin{equation}
					A_n=\sqrt n(\hat\tau^{M_1}-\tau_{*,n}^{M_1}),
					\qquad
					B_n=n^{-1/2}\{LR_n-n\mu_{*,n}\}.
					\label{eq:joint_post_statistics}
				\end{equation}
				Because the two statistics are computed from the same observations, they are generally correlated.
				
				Let $\psi_{i,n}$ denote the version of the influence function
				\(\psi_i^{M_1}\) in \eqref{eq:functional_if} under the \(n\)th
				data-generating law, and let
				\[
				\phi_{i,n}
				=\ell_i^{M_1}(\theta_{*,n}^{M_1})
				-\ell_i^{M_2}(\theta_{*,n}^{M_2})
				-\mu_{*,n}.
				\]
				The standard score and delta-method expansions give
				\[
				A_n=\frac{1}{\sqrt n}\sum_{i=1}^n\psi_{i,n}+o_p(1).
				\]
				A second-order expansion of each maximized log likelihood around its
				pseudo-true parameter shows that parameter estimation changes the
				unscaled likelihood ratio by $O_p(1)$. Therefore,
				\[
				B_n=\frac{1}{\sqrt n}\sum_{i=1}^n\phi_{i,n}+o_p(1).
				\]
				Under Assumption~\ref{assu:regularity} and the standard triangular-array Lindeberg condition for $(\psi_{i,n},\phi_{i,n})$, the multivariate Lindeberg--Feller theorem gives
				\begin{equation}
					\begin{pmatrix}A_n\\ B_n\end{pmatrix}
					\Longrightarrow
					N\!\left\{
					\begin{pmatrix}0\\0\end{pmatrix},
					\begin{pmatrix}V_1&c\\c&\omega^2\end{pmatrix}
					\right\},
					\qquad
					V_1>0,\quad \omega^2>0,\quad c^2<V_1\omega^2.
					\label{eq:joint_post_expansion}
				\end{equation}
				The covariance $c$ captures the first-order dependence between functional estimation and model selection. In practice, set
				\[
				\hat\psi_i
				=h_{\hat\theta^{M_1}}(X_i)-\hat\tau^{M_1}
				-\hat d_1^{\T}\hat A_{M_1}^{-1}\hat S_{M_1,i},
				\qquad
				\hat\phi_i=m_i-\bar m_n,
				\]
				where
				\[
				\hat d_1=\frac1n\sum_{i=1}^n
				\nabla_\theta h_{\hat\theta^{M_1}}(X_i),
				\qquad
				\overline{\hat\psi}=\frac1n\sum_{i=1}^n\hat\psi_i,
				\]
				and estimate the covariance matrix in \eqref{eq:joint_post_expansion} by
				\[
				\hat\Sigma
				=\frac1n\sum_{i=1}^n
				\begin{pmatrix}\hat\psi_i-\overline{\hat\psi}\\ \hat\phi_i\end{pmatrix}
				\begin{pmatrix}\hat\psi_i-\overline{\hat\psi}\\ \hat\phi_i\end{pmatrix}^{\!\T}.
				\]
				Write $\hat V_1$ and $\hat c$ for the first diagonal and off-diagonal
				entries of $\hat\Sigma$; its second diagonal entry is
				$\hat\omega_n^2$.
				The law of large numbers and continuous mapping theorem give
				consistency for the covariance matrix displayed in
				\eqref{eq:joint_post_expansion}.
				
				\subsection{Proof of Corollary~\ref{cor:selection_determinate}}
				
				\begin{proof}
					{Under $\omega_*^2>0$, Lemma~\ref{lem:variance_consistency}(i) gives $\hat\omega_n^2\to\omega_*^2$ almost surely, and hence the first-stage statistic satisfies}
					\[
					\hat T_n=n\hat\omega_n^2\longrightarrow\infty
					\]
					in probability. Hence the first-stage test rejects indistinguishability
					with probability tending to one.
					
					Next, by \eqref{eq:lr_consist},
					\[
					\frac{1}{n}
					LR_n(\hat\theta^{M_1},\hat\theta^{M_2})
					\xrightarrow{a.s.}
					\mu_*,
					\]
					and by Lemma~\ref{lem:variance_consistency}(i),
					\[
					\hat\omega_n\xrightarrow{a.s.}\omega_*>0.
					\]
					Therefore
					\[
					\hat V_n
					=
					\frac{
						LR_n(\hat\theta^{M_1},\hat\theta^{M_2})
					}{
						\sqrt n\,\hat\omega_n
					}
					=
					\sqrt n\,\frac{\mu_*}{\omega_*}+o_p(\sqrt n).
					\]
					Thus, if $\mu_*>0$, then $\hat V_n\to+\infty$ in probability, while
					if $\mu_*<0$, then $\hat V_n\to-\infty$ in probability. Since the
					second-stage rule selects $M_1$ when $\hat V_n>q_2$, and the
					first-stage rejection probability tends to one, it follows that
					\[
					\mu_*>0
					\quad\Longrightarrow\quad
					P(\widehat M=M_1)\to1,
					\]
					and
					\[
					\mu_*<0
					\quad\Longrightarrow\quad
					P(\widehat M=M_1)\to0.
					\]
				\end{proof}
				
				\subsection{Proof of Theorem~\ref{thm:post_selection_wald}}
				
				\begin{proof}
					Let
					\[
					S_n=\{\widehat M=M_1\}.
					\]
					By Corollary~\ref{cor:selection_determinate}, under $\mu_*>0$,
					\[
					p_n:=P(S_n)\longrightarrow1.
					\]
					Hence $p_n>0$ for all sufficiently large $n$.
					
					We first show that conditioning on $S_n$ is asymptotically negligible.
					For any measurable event $E_n$,
					\begin{align*}
						\left|
						P(E_n\mid S_n)-P(E_n)
						\right|
						&=
						\left|
						\frac{P(E_n\cap S_n)}{p_n}
						-
						P(E_n)
						\right| \\
						&=
						\left|
						\frac{P(E_n)-P(E_n\cap S_n^c)}{p_n}
						-
						P(E_n)
						\right| \\
						&\le
						\left|
						\frac{P(E_n)}{p_n}
						-
						P(E_n)
						\right|
						+
						\frac{P(E_n\cap S_n^c)}{p_n} \\
						&\le
						\frac{1-p_n}{p_n}
						+
						\frac{P(S_n^c)}{p_n} \\
						&=
						\frac{2(1-p_n)}{p_n}
						\longrightarrow0.
					\end{align*}
					Thus conditional and unconditional probabilities have the same
					first-order limit whenever the conditioning event has probability
					tending to one. Under $\mu_*>0$, $S_n$ is asymptotically certain rather
					than a nondegenerate truncation event.
					
					Now define
					\[
					A_n=
					\sqrt n(\hat\tau^{M_1}-\tau_*^{M_1}).
					\]
					By Theorem~\ref{thm:ate_normality},
					\[
					A_n
					\xrightarrow{d}
					N(0,V_{M_1})
					\]
					unconditionally. Let $x$ be any continuity point of the limiting normal
					distribution, and set
					\[
					E_n(x)=\{A_n\le x\}.
					\]
					Applying the preceding bound to $E_n(x)$ gives
					\[
					\left|
					P(A_n\le x\mid S_n)
					-
					P(A_n\le x)
					\right|
					\longrightarrow0.
					\]
					Since
					\[
					P(A_n\le x)
					\longrightarrow
					\Phi\left(\frac{x}{\sqrt{V_{M_1}}}\right),
					\]
					we obtain
					\[
					P(A_n\le x\mid S_n)
					\longrightarrow
					\Phi\left(\frac{x}{\sqrt{V_{M_1}}}\right).
					\]
					Therefore,
					\[
					\sqrt n(\hat\tau^{M_1}-\tau_*^{M_1})
					\;\Big|\;S_n
					\xrightarrow{d}
					N(0,V_{M_1}).
					\]
					
					It remains to prove the coverage statement. Let
					\[
					E_n
					=
					\left\{
					\tau_*^{M_1}\in\mathcal C_{1-\alpha}^{\mathrm{Wald}}
					\right\}.
					\]
					By Theorem~\ref{thm:ate_normality}, $\hat V_{M_1}\to_p V_{M_1}$, and
					Slutsky's theorem,
					\[
					P(E_n)\longrightarrow1-\alpha.
					\]
					Applying the conditioning bound to this coverage event gives
					\[
					\left|
					P(E_n\mid S_n)-P(E_n)
					\right|
					\longrightarrow0.
					\]
					Consequently,
					\[
					P\left\{
					\tau_*^{M_1}\in\mathcal C_{1-\alpha}^{\mathrm{Wald}}
					\,\middle|\,
					\widehat M=M_1
					\right\}
					\longrightarrow
					1-\alpha.
					\]
				\end{proof}
				\subsection{Boundary case and selection-normal confidence intervals}
				\label{app:selection_normal_boundary}
				
				Theorem~\ref{thm:post_selection_wald} assumes a fixed data-generating law
				with $\mu_*>0$. In that case
				\[
				P(\widehat M=M_1)\to1,
				\]
				so the conditioning event is asymptotically certain and no first-order
				selective correction is needed. The situation is different for a
				sequence satisfying $\sqrt n\,\mu_{*,n}\to\delta$. The selection event
				then has a nondegenerate limiting probability and becomes a genuine
				truncation event. Ordinary Wald inference generally needs to be replaced
				by inference based on the selection-normal limit.
				
				Specifically, let $\mu_{*,n}$, $\theta_{*,n}^{M_1}$,
				$\theta_{*,n}^{M_2}$, and $\tau_{*,n}^{M_1}$ denote the versions of
				$\mu_*$, $\theta_*^{M_1}$, $\theta_*^{M_2}$, and $\tau_*^{M_1}$ under
				the $n$th data-generating law. Consider a sequence satisfying
				\[
				\sqrt n\,\mu_{*,n}\longrightarrow \delta,
				\qquad \delta\in\mathbb R.
				\]
				Let $A_n$ and $B_n$ be the statistics defined in \eqref{eq:joint_post_statistics}. Their joint Gaussian limit is given in \eqref{eq:joint_post_expansion}.
				Let
				\[
				Z_n=\frac{B_n}{\omega},
				\qquad
				\rho=\frac{c}{\sqrt{V_1}\,\omega},
				\qquad
				a_\delta=q_2-\frac{\delta}{\omega}.
				\]
				The covariance condition above ensures $|\rho|<1$, so the Gaussian
				limit is nondegenerate.
				By the definition of $B_n$, the second-stage statistic satisfies the
				exact identity
				\[
				\hat V_n
				=
				\frac{B_n+\sqrt n\,\mu_{*,n}}{\hat\omega_n}
				=
				\frac{\omega}{\hat\omega_n}Z_n
				+
				\frac{\sqrt n\,\mu_{*,n}}{\hat\omega_n}.
				\]
				Hence, after the first-stage rejection event whose probability tends to
				one, selection of $M_1$ is the event
				\[
				\left\{
				Z_n>
				\frac{\hat\omega_n}{\omega}q_2
				-
				\frac{\sqrt n\,\mu_{*,n}}{\omega}
				\right\}.
				\]
				The random threshold converges in probability to
				$a_\delta=q_2-\delta/\omega$. Because the limiting variable $Z$ has a
				continuous density, this event is asymptotically equivalent in
				probability to $\{Z_n>a_\delta\}$.
				Let $(A,Z)$ denote the Gaussian limit of $(A_n,Z_n)$. Since
				$P(Z>a_\delta)=1-\Phi(a_\delta)>0$ and the boundary
				$\{Z=a_\delta\}$ has probability zero, conditional weak convergence
				gives
				\[
				A_n\mid\{\widehat M=M_1\}
				\xrightarrow{d}
				A\mid\{Z>a_\delta\}.
				\]
				
				For any $x\in\mathbb R$, the numerator of the relevant conditional
				probability is
				\[
				P(A\le x,Z>a_\delta)
				=
				P(A\le x)-P(A\le x,Z\le a_\delta).
				\]
				The conditional distribution function of the limiting post-selection
				law is therefore
				\begin{equation}
					F_\delta(x)
					=
					P(A\le x\mid Z>a_\delta)
					=
					\frac{
						\Phi(x/\sqrt{V_1})
						-
						\Phi_2(x/\sqrt{V_1},a_\delta;\rho)
					}{
						1-\Phi(a_\delta)
					},
					\label{eq:app_selection_normal_cdf}
				\end{equation}
				where $\Phi$ is the standard normal distribution function and
				$\Phi_2(\cdot,\cdot;\rho)$ is the bivariate standard normal distribution
				function with correlation $\rho$. The law in
				\eqref{eq:app_selection_normal_cdf} is a selection-normal, or
				skew-normal-type, limit. It is generally not the ordinary $N(0,V_1)$
				law because conditioning on $Z>a_\delta$ changes the distribution of
				$A$ whenever $A$ and $Z$ are correlated.
				
				Let
				\[
				v_{p,\delta}=F_\delta^{-1}(p),
				\qquad 0<p<1,
				\]
				denote the $p$th quantile of the selection-normal distribution in
				\eqref{eq:app_selection_normal_cdf}. If $\delta$ were
				known, an oracle selective confidence interval for
				$\tau_{*,n}^{M_1}$ would be
				\begin{equation}
					\mathcal C_{1-\alpha}^{\mathrm{sel}}(\delta)
					=
					\left[
					\hat\tau^{M_1}
					-
					\frac{v_{1-\alpha/2,\delta}}{\sqrt n},
					\;
					\hat\tau^{M_1}
					-
					\frac{v_{\alpha/2,\delta}}{\sqrt n}
					\right].
					\label{eq:app_oracle_selection_ci}
				\end{equation}
				For feasible calculation, let $\widehat F_\delta$ be the distribution
				function in \eqref{eq:app_selection_normal_cdf} with
				$V_1$, $c$, and $\omega$ replaced by their consistent estimators and
				with
				\[
				\hat\rho=\frac{\hat c}{\sqrt{\hat V_1}\,\hat\omega_n},
				\qquad
				\hat a_\delta=q_2-\frac{\delta}{\hat\omega_n}.
				\]
				Set $\hat v_{p,\delta}=\widehat F_\delta^{-1}(p)$.
				The plug-in oracle interval is then
				\begin{equation}
					\widehat{\mathcal C}_{1-\alpha}^{\mathrm{sel}}(\delta)
					=
					\left[
					\hat\tau^{M_1}
					-
					\frac{\hat v_{1-\alpha/2,\delta}}{\sqrt n},
					\;
					\hat\tau^{M_1}
					-
					\frac{\hat v_{\alpha/2,\delta}}{\sqrt n}
					\right].
					\label{eq:app_plugin_selection_ci}
				\end{equation}
				
				This interval is asymmetric in general. The asymmetry comes from the
				fact that the conditional distribution of $A$ given $Z>a_\delta$ need
				not be centered normal. Let $\varphi$ denote the standard normal
				density. If
				\[
				\lambda(a)=\frac{\varphi(a)}{1-\Phi(a)},
				\qquad
				\kappa(a)=1+a\lambda(a)-\lambda(a)^2,
				\]
				then
				\[
				E(A\mid Z>a_\delta)
				=
				\frac{c}{\omega}\lambda(a_\delta),
				\]
				and
				\[
				\operatorname{Var}(A\mid Z>a_\delta)
				=
				V_1-\frac{c^2}{\omega^2}\{1-\kappa(a_\delta)\}.
				\]
				Thus the conditional limit is centered at zero only when
				\[
				c=0
				\]
				or in special limiting cases where the truncation becomes irrelevant.
				When $c=0$, $A$ and $Z$ are independent in the Gaussian limit, and hence
				conditioning on $Z>a_\delta$ does not change the law of $A$. Otherwise,
				the post-selection distribution is generally shifted and asymmetric.
				
				The oracle interval in \eqref{eq:app_oracle_selection_ci} satisfies
				\[
				P\left\{
				\tau_{*,n}^{M_1}
				\in
				\mathcal C_{1-\alpha}^{\mathrm{sel}}(\delta)
				\,\middle|\,
				\widehat M=M_1
				\right\}
				\longrightarrow
				1-\alpha,
				\]
				provided the interval is constructed at the value $\delta$ governing
				$\sqrt n\,\mu_{*,n}\to\delta$. The qualifier is important: unlike the
				case $\mu_*>0$ in Theorem~\ref{thm:post_selection_wald}, the condition
				$\sqrt n\,\mu_{*,n}\to\delta$ introduces the additional parameter $\delta$. Since
				$\delta$ is generally not consistently estimable under
				{$\sqrt n\,\mu_{*,n}\to\delta$, a feasible interval must account for uncertainty about $\delta$, for example through a confidence set and a least-favourable envelope.}
				
				{We use the latter idea to obtain a feasible interval. Let
					\[
					U_n=\frac{LR_n(\hat\theta^{M_1},\hat\theta^{M_2})}{\sqrt n}.
					\]
					For a candidate drift $d$, the plug-in conditional distribution function of $U_n$, given selection of $M_1$, is
					\begin{equation}
						\widehat H_d(u)
						=
						\frac{
							\Phi\{(u-d)/\hat\omega_n\}
							-\Phi\{q_2-d/\hat\omega_n\}
						}{
							1-\Phi\{q_2-d/\hat\omega_n\}
						},
						\qquad u>q_2\hat\omega_n.
						\label{eq:drift_conditional_cdf}
					\end{equation}
					At the true drift, $\widehat H_\delta(U_n)$ is asymptotically uniform conditional on selecting $M_1$. Hence, for $0<\beta<\alpha$,
					\begin{equation}
						\mathcal D_{n,1-\beta}
						=\left\{
						d:\frac{\beta}{2}\leq \widehat H_d(U_n)
						\leq 1-\frac{\beta}{2}
						\right\}
						\label{eq:drift_confidence_set}
					\end{equation}
					is a conditional $(1-\beta)$ confidence set for $\delta$.}
				
				For $0<\eta<1$, write
				$\widehat{\mathcal C}_{1-\eta}^{\mathrm{sel}}(d)
				=[\ell_{n,\eta}(d),u_{n,\eta}(d)]$
				for the plug-in oracle interval in \eqref{eq:app_plugin_selection_ci}, with $\alpha$ there replaced by $\eta$. The least-favourable envelope is
				\begin{equation}
					\widehat{\mathcal C}_{1-\alpha}^{\mathrm{LF}}
					=
					\left[
					\inf_{d\in\mathcal D_{n,1-\beta}}\ell_{n,\alpha-\beta}(d),
					\quad
					\sup_{d\in\mathcal D_{n,1-\beta}}u_{n,\alpha-\beta}(d)
					\right].
					\label{eq:least_favourable_ci}
				\end{equation}
				If the true drift belongs to $\mathcal D_{n,1-\beta}$, the oracle interval at that drift is contained in the envelope. A union bound therefore gives
				\[
				\liminf_{n\to\infty}
				P\!\left\{
				\tau_{*,n}^{M_1}\in
				\widehat{\mathcal C}_{1-\alpha}^{\mathrm{LF}}
				\,\middle|\,\widehat M=M_1
				\right\}
				\geq 1-\alpha.
				\]
				No independence between the drift confidence set and the oracle interval is required.
				
				{Finally, let $\kappa_n\to\infty$ with $\kappa_n/\sqrt n\to0$ and define, after selecting $M_1$,
					\begin{equation}
						\widehat{\mathcal C}_{1-\alpha}^{\mathrm{ela}}
						=
						\begin{cases}
							\widehat{\mathcal C}_{1-\alpha}^{\mathrm{LF}},
							& q_2<\hat V_n\leq\kappa_n,\\[2pt]
							\mathcal C_{1-\alpha}^{\mathrm{Wald}},
							& \hat V_n>\kappa_n.
						\end{cases}
						\label{eq:elastic_ci}
					\end{equation}
					When $\sqrt n\,\mu_{*,n}\to\delta$, $\hat V_n=O_p(1)$ and the elastic interval uses the least-favourable branch with probability tending to one. When $\mu_*>0$ under a fixed data-generating law, $\hat V_n$ diverges at rate $\sqrt n$ and the Wald branch is used with probability tending to one. Thus the elastic interval has pointwise conditional coverage in both cases, while avoiding an unnecessary selective correction when selection is asymptotically determinate. The construction for selection of $M_2$ follows by reversing the likelihood contrast and the truncation direction.}
				
				This explains why Theorem~\ref{thm:post_selection_wald} has an ordinary
				normal limit and uses the usual Wald interval when $\mu_*>0$, whereas
				the case $\sqrt n\,\mu_{*,n}\to\delta$ has a
				selection-normal limit and requires selection-adjusted, generally
				asymmetric confidence intervals.
\end{appendix}

\end{document}